\documentclass[12pt]{amsart}
\usepackage{amsmath, amsthm, amscd, tikz, amssymb, amsfonts, amsxtra, amssymb, latexsym}
\usepackage{enumerate}
\usepackage{verbatim}
\usepackage[normalem]{ulem} 
\usepackage{array}
\usepackage[utf8]{inputenc}
\usepackage{hyperref}
\hypersetup{colorlinks = true,	allcolors  = blue}

\def\fq{\mathbb{F}_{q}}
\def\fqs{\mathbb{F}_{q^2}}
\def\fqc{\mathbb{F}_{q^3}}

\def\cF{\mathcal{F}}

\def\cX{\mathcal{X}}
\def\cC{\mathcal{C}}
\def\cF{\mathcal{F}}
\def\cP{\mathcal{P}}

\def\cH{\mathcal{H}}
\def\cL{\mathcal{L}}
\def\cK{\mathcal{K}}
\def\cX{\mathcal{X}}
\def\N{\mathbb{N}}
\def\Z{\mathbb{Z}}

\def\cF{\mathcal{F}}
 \def\cM{\mathcal{M}}
\newcommand{\con}{\operatorname{Con}}
\newcommand{\Cot}{\operatorname{Cotr}}

\newcommand{\rr}{\mathcal{L}}

\newcommand{\tr}{\operatorname{Tr}}

\def\a{\alpha}

\def\b{\beta}
\def\g{\gamma}
\def\ff{\mathbb{F}}
\def\bx{\bf x}

\DeclareMathOperator\supp{supp}
\DeclareMathOperator\diff{Diff}

\newcommand{\sk}{\smallskip}
\newcommand{\msk}{\medskip}

\newtheorem{theorem}{Theorem}[section]

\newtheorem{corollary}[theorem]{Corollary}
\newtheorem{proposition}[theorem]{Proposition}

\theoremstyle{definition}
\newtheorem{definition}[theorem]{Definition}
\newtheorem{example}[theorem]{Example}
\newtheorem{remark}[theorem]{Remark}

\begin{document} \sloppy
	\numberwithin{equation}{section}
	\title{Lifting iso-dual algebraic geometry codes} % on extensions \\ of function fields} 
	\author[M.\@ Chara, R.\@ Podestá, L.\@ Quoos, R.\@ Toledano]{María Chara, Ricardo Podestá, Luciane Quoos, Ricardo Toledano}
	\dedicatory{\today}
	\keywords{Isodual codes, AG-codes, algebraic function field, tower of function fields} 
	\thanks{2020 {\it Mathematics Subject Classification.} Primary 11T71, 14G50, 94B05, 94B27.} % 14Q05.}
	\thanks{The first, second and fourth authors are partially supported by CONICET, FONCyT, SECyT-UNC, and CAI+D-UNL. The third author was financed in part by the Coordenação de Aperfeiçoamento de Pessoal de Nível Superior, CAPES MATH AMSUD 88881.647739/2021-01 }	
	\address{María Chara, Researcher of CONICET at Facultad de Ingeniería Química, Universidad Nacional del Litoral, Santiago del Estero 2829, (3000) Santa Fe, Argentina. 
		\newline {\it E-mail: mchara@santafe-conicet.gov.ar}}
	
	\address{Ricardo Podest\'a, FaMAF -- CIEM (CONICET), Universidad Nacional de C\'ordoba, \newline
	Av.\@ Medina Allende 2144, Ciudad Universitaria, (5000) C\'ordoba, Argentina. 
	\newline {\it E-mail: podesta@famaf.unc.edu.ar}}

	\address{Luciane Quoos, Universidade Federal do Rio de Janeiro, Centro de Tecnologia, Cidade Universitária, Av.\@ Athos da Silveira Ramos 149, Ilha do Fundão, CEP~21.941-909, Brazil. {\it E-mail: luciane@im.ufrj.br}}

	\address{Ricardo Toledano, Facultad de Ingeniería Química, Universidad Nacional del Litoral, Santiago del Estero 2829, (3000) Santa Fe, Argentina. {\it E-mail: ridatole@gmail.com}}

\begin{abstract}
	In this work we investigate the problem of producing iso-dual algebraic geometry (AG) codes over a finite field $\fq$ with $q$ elements. Given a finite separable extension $\cM/\cF$ of function fields  and an iso-dual AG-code $\cC$ defined over $\cF$, we provide a general method to lift the code $\cC$ to another iso-dual AG-code $\tilde \cC$ defined over $\cM$ under some assumptions on the divisors $D$ and $G$ and on the parity of the involved different exponents. We apply this method to lift iso-dual AG-codes over the rational function field to  elementary abelian $p$-extensions, like the maximal function fields defined by the Hermitian, Suzuki, and one covered by the $GGS$ function field.  We also obtain long binary and ternary iso-dual AG-codes defined over cyclotomic extensions.
\end{abstract}

\maketitle

\section{Introduction} \label{sec:intro}
Let $\fq$ be a finite field with $q$ elements. A linear code $\cC$ over $\fq$ is an $\fq$-linear subspace of $\fq^n$ for $ n\geq 1$.  Associated to a code   $\cC \subseteq \fq^n$ we have three parameters: its length $n$, its dimension $k$ as a vector space over $\fq$ and its minimum distance $d$ (Hamming distance).
In the 1980's, using concepts and tools coming from algebraic geometry, Goppa constructed error-correcting linear codes from function fields defined over a finite field, see \cite{G1981} and \cite{G1982}. They are called algebraic geometry (AG) codes and have played an important role in the theory of error-correcting codes. They were used to improve the Gilbert–Varshamov bound about the limit of the parameters of a code \cite{TVZ1982}, and this was a remarkable result at that time. Moreover, every linear code can be realized as an algebraic geometry code \cite{PSV1991}.  

Let us recall first the definition of an algebraic geometry code. We use the language of function fields over finite fields following \cite{Stichbook09}. For  a function field $\cF/\fq$, consider the divisor $D=P_1+\dots+P_n$ given by the sum of pairwise distinct rational places of $\cF$,  and another divisor $G$ such that $P_i$ is not in the support of $G$ for $i = 1,\dots ,n$. The linear {\it algebraic geometry code} $C_{\mathcal{L}}(D, G)$ defined over $\cF$ is given by
\begin{equation} \label{defAGcode}
	C_{\mathcal{L}}(D, G)=\{ (f(P_1),\ldots, f(P_n))\, :\, f\in \mathcal{L}(G)\} \subseteq \fq^n,
\end{equation}
where $\cL(G)=\{z\in \cF\,:\, (z)\ge -G\}\cup \{0\}$ denotes the Riemann-Roch space associated to the divisor $G$.

Recall now that the dual code $\cC^\perp$ of a linear code $\cC$ is the orthogonal complement of $\cC$ in $\fq^n$ with respect to the standard inner product of $\fq^n$. A code is said to be {\em self-dual} if $\cC=\cC^\perp$.
Self-dual codes have been investigated in \cite{MST2008}, they have applications to quantum codes through a construction in \cite{KM2008} (see also \cite{BMZ2021}, \cite{LP2017} and  \cite{MTF2016}); and new constructions have been recently proposed in, for example, \cite{S2021} and \cite{SYS2023}. In \cite{BS2019}, infinite families of self-dual codes which are asymptotically better than the asymptotic Gilbert–Varshamov bound were constructed.

The condition of being self-dual can be relaxed by considering the following notion of equivalence of linear codes: two linear codes $\cC_1$ and $\cC_2$ are said to be {\it equivalent} if there exists a vector ${\bf x}\in (\mathbb{F}_q^*)^n$ such that $\cC_1={\bf x}\cdot \cC_2$ where for  ${\bf x}=(x_1,\ldots,x_n)$ we denote $${\bf x}\cdot \cC =\{(x_1·c_1,\ldots,x_n·c_n): (c_1,\ldots,c_n)\in \cC \}.$$ Now a linear code $\cC$ is called {\it iso-dual} if it is equivalent to its dual code $\cC^\perp$, that is if there exists  a non-zero vector ${\bf x}\in (\mathbb{F}_q^*)^n$ such that 
	$$\cC^\perp = {\bf x} \cdot \cC.$$
We will speak of ${\bf x}$-iso-dual codes when mentioning the vector ${\bf x}$ explicitly is needed.

The main goal of this work is to investigate the construction of iso-dual AG codes. They were first studied in full generality in \cite{Stich1988}, where several concrete examples of iso-dual AG codes over function fields of arbitrary genus over a finite field were presented. Much later it was proved in \cite{Stich06} that the class of iso-dual codes attains the Tsfasman-Vladut and Zink bound over a finite field of quadratic cardinality. Iso-dual AG codes also showed up in the study of the so called order bounds for the minimum distance of AG-codes (see, for instance, \cite{GMRT2011} and the references therein). Overall the construction of iso-dual codes in a function field poses a significant challenge, as it relies on a deep understanding of differentials and function divisors possessing specific properties (see Proposition \ref{prop:isodual}). We propose here an alternative method to construct them (see Theorem \ref{teolevantado}): given a finite separable extension $\cM/\cF$ of function fields and an iso-dual AG-code $\cC$ defined over $\cF$,  we use the conorm map (see Section \ref{sec:preliminaries}) to lift the AG-code $\cC$ to another iso-dual AG-code over $\cM$. We obtain in this way a longer iso-dual AG-code whose parameters can be estimated in the standard way (see  Corollary \ref{parameters}) in many cases. 
To the best of our knowledge, iso-dual codes have not undergone a thorough investigation.

We employ this method of lifting iso-dual AG-codes across various scenarios. The most favourable of them is when $\cF$ is a rational function field, because iso-dual codes over a rational function field can be easily constructed (see items \eqref{isorat} and \eqref{isoratcan} of Proposition \ref{prop:isodual}).   In particular with our method we can construct iso-dual codes over {\it maximal function fields}, that is, function fields defined by algebraic curves $\cX$ of genus $g(\cX)$ such that its number $\cX(\fqs)$ of $\fqs$-rational points attains the Hasse-Weil upper bound
	$$ \# \cX(\fqs) = q^2+1+2g(\cX)q.$$
Maximal curves are a fruitful ambient for explicit constructions of codes, since its is well know that  curves with a small genus with respect its number of rational points, produce codes with better relative  parameters. Research in codes in recent years has proven to be a highly productive field of investigation. For instance, construction of locally recoverable codes over maximal curves can be found in \cite{BMQ2020} or \cite{CKMTW}, AG-codes over the maximal Beelen-Montanucci curve \cite{VL2022}, weights of dual codes over the maximal GK-curve \cite{BB2023}, lifting of AG-codes in \cite{CPT2023}. The  maximal Hermitian curve was used for  applications to quantum codes \cite{FWF2024},   self orthogonal maximum distance separable codes \cite{LSZLS2023},  investigation of the isometry-dual property in flags of codes \cite{BCQ2022}, construction of codes using places of higher degree \cite{MM2005}, and many point codes in \cite{KNT2020}.

The paper is organized as follows. In Section \ref{sec:preliminaries} we present the necessary background on the theory of function fields and codes. We recall some definitions and  basic facts on AG-codes and iso-dual AG-codes. In Section \ref{sec:isoHerm}, we construct families of iso-dual and self-dual AG-codes $C_\cL(D,G)$ over the Hermitian function field $\cH/\ff_{q^2}$ 
(see Theorem~\ref{thm:herm}). 
In Section \ref{sec:lifting}, given a finite separable extension $\cM/\cF$ of algebraic function fields we propose a way to lift an iso-dual code over $\cF$ to an iso-dual code over $\cM$, given certain constraints in the extension (see Theorem \ref{teolevantado}) and the involved divisors of $\cF$. 
In Corollary \ref{parameters} we give the parameters of the lifted code. In Section~\ref{sec: from rational} we consider elementary abelian $p$-extensions of the rational function field. By lifting iso-dual rational AG-codes we obtain iso-dual AG-codes over the Hermitian,  Suzuki and GGS curves.
In Section \ref{sec: from hermitian} we consider a new curve $\cX$ obtained as an extension of the Hermitian function field.  We compute its genus and its number of rational points and in Theorem \ref{thm: isodual Herm} we show how to obtain iso-dual AG-codes defined over $\cX$ using our method of lifting iso-dual codes. 	
Finally, in Section \ref{sec:cyclotomic} we consider cyclotomic extensions. By lifting very simple iso-dual codes AG-codes over binary and ternary rational function fields, we get long binary and ternary iso-dual AG-codes over subfields of cyclotomic extensions (see Theorems \ref{thm:binary} and \ref{thm:ternary}).

\section{Preliminaries} \label{sec:preliminaries}
Here we recall some basic facts of extensions of algebraic function fields, divisors, AG-codes and iso-dual AG-codes.

\subsection*{Algebraic function fields} 
Let $\cF/\fq$ be an algebraic function field in one variable of genus $g=g(\cF)$. We denote by $\mathcal P_\cF$ the set of places in $\cF$, by $\Omega_\cF$ the space of Weil differentials in $\cF$, by $v_{P}$ the discrete valuation of $\cF/\fq$ associated to the place $P\in \cP_\cF$, and by $\mathrm{Div}(\cF)$ the free abelian group generated by the places in $\cF$. An element in $\mathrm{Div}(\cF)$ is called a divisor. For a function $z \in \cF$ we let $(z)^\cF, (z)^\cF_\infty$ and $(z)^\cF_0$ stand for the principal, pole and zero divisors of the function $z$ in $\cF$, respectively. Two divisors $A,B \in \mathrm{Div}(\cF)$ are equivalent, denoted $A \sim B$, if they differ by a principal divisor, i.e.\@ $B=A+(z)^\cF$ for some $z \in \cF$.

Let us consider now a separable function field extension $\cF'/\cF$. The conorm of a place $P$ in $\cF$ is the divisor in $\cF'$ defined by
	$$ \con_{\cF' /\cF} (P) = \sum_{Q \in \cF', Q|P} e(Q|P) \, Q,$$
where $e(Q|P)$ denotes the ramification index of $Q$ over $P$.
For a divisor $D$ in $\mathrm{Div}(\cF)$ the \textit{conorm} map is the natural homomorphism from $\mathrm{Div}(\cF)$ to $\mathrm{Div}(\cF')$ 
extended by linearity; 
that is, 
\begin{equation} \label{eq:conorm}
	\con_{\cF' /\cF} \big(\sum_i n_i P_i \big) = \sum_i n_i \, \con_{\cF' /\cF} (P_i).
\end{equation}

Suppose that $\cF'$ and $\cF$ have fields of constants $K'$ and $K$, respectively (in the paper, however, we will only consider geometric extensions of functions fields, i.e.\@ with $K'=K$). From Theorem~3.4.6 in \cite{Stichbook09} we have that for every Weil differential $\omega$ of $\cF/K$ there exists a unique Weil differential $\omega'$ in $\cF'/K'$ such that $\tr_{K'/K}(\omega'(\alpha)) = \omega(\tr_{F'/F}(\alpha))$ for every $\alpha$ in the adele space 
$\mathcal{A}_{F'/F} = \{\alpha \in \mathcal{A}_{\cF'} : \alpha_{P'}=\alpha_{Q'} \text{ whenever } P'\cap F = Q' \cap F\}$. This Weil differential is called the \textit{cotrace} of $\omega$ in $\cF'/\cF$ and it is denoted by $\Cot_{\cF'/\cF}(\omega)$. 
If $\omega \ne 0$ and $(\omega)$ is the corresponding divisor in $\cF$, one has 
\begin{equation} \label{eq:cotrace}
	(\Cot_{\cF'/\cF}(\omega)) = \con_{\cF' /\cF}((\omega)) + \diff(\cF'/\cF),
\end{equation}
with $\diff(\cF'/\cF)$ the different divisor of $\cF'/\cF$ defined by 
\begin{equation} \label{eq:diff}
	\diff(\cF'/\cF) =\sum_{P \in P_\cF} \sum_{P'|P} d(P'|P) \, P
\end{equation}
where $d(P'|P)\ge 0$ is the different exponent of $P'|P$ (which is $0$ for almost all places).
We recall that if $\cF'/\cF$ is also finite and $g'$ and $g$ denote the genera of $\cF'$ and $\cF$ respectively, then
the Riemann-Hurwitz genus formula asserts that 
\begin{equation} \label{eq: genus fla}
	2g'-2=[\cF':\cF](2g-2) + \deg(\diff(\cF'/\cF))
\end{equation}
(see for instance Theorem 3.4.13 in \cite{Stichbook09}).

\subsection*{Iso-dual AG-codes} \label{sec:isodual}
A linear code $\cC$ is an $\fq$-linear subspace of $\fq^n$ with $n\geq 1$. Associated to a code we have three important parameters, its length $n$,  its dimension $k$ as a vector subspace over $\fq$, and its minimum (Hamming) distance $d$. One says that $\cC$ is an $[n, k, d]$-code. From now on we will use AG-codes $C_\cL(D,G)$ over an algebraic function field $\cF/\ff_q$ as defined in \eqref{defAGcode}.

We now recall the basic bounds for the parameters of an AG-code.

\begin{proposition}%[parameters] %Theorem 2.2.2, Corollary 2.2.3 in \cite{Stichbook09}] 
	\label{prop:dimcode} 
The AG-code $C_\cL(D, G)$ as in \eqref{defAGcode} is an $[n,k,d]_q$-code with  
\begin{equation} \label{dimcode} 
	k=\ell(G)-\ell(G-D) \qquad \text{ and } \qquad d \geq n-\deg(G).
\end{equation}
Moreover, we have that:
\begin{enumerate}[$(a)$]
\item If $\deg (G) < n$, then $k=\ell(G) \geq \deg(G)+1-g$. \msk 

\item If $2g-2<\deg(G) < n$, then $k=\ell(G)= \deg(G)+1-g$.
\end{enumerate}
\end{proposition}

\begin{proof}
	See for instance Theorem 2.2.2 and Corollary 2.2.3 in \cite{Stichbook09}.
\end{proof}

Let $\cC \subseteq \fq^n$ be a linear code and ${\bf x} =(x_1, \ldots, x_n) \in (\fq^\ast)^n$. 
The set
	$${\bf x} \cdot \cC :=\{ (x_1 c_1, \ldots, x_n c_n) : (c_1, \ldots, c_n) \in \cC \},$$
is clearly another linear code over $\fq$. We notice that the codes $\cC$ and ${\bf x} \cdot \cC$ have the same length, dimension and minimum distance. 
%Following Definition 2.2.13 of \cite{Stichbook09} we define equivalent codes.
We use the same notion of equivalence as in Definition 2.2.13 of \cite{Stichbook09}. 
\begin{definition}
Two linear codes $\cC_1$ and $\cC_2$ over $\fq$ are \textit{equivalent} if $\cC_2={\bf x} \cdot\cC_1$ for some ${\bf x} \in (\fq^\ast)^n$. In this case we write $\cC_1 \sim_{\bf x} \cC_2$, or simply $\cC_1 \sim \cC_2$ when the distinction of the vector $\bf x$ is not necessary.
%to express that $\cC_1$ and $\cC_2$ are equivalent. 
\end{definition}

In the next proposition we collect well-known results on the equivalence of AG-codes.

\goodbreak 
\begin{proposition} \label{prop2210}
%[Corollaries 2.7, 3.6 in \cite{Stich1988}, Propositions 2.2.10, 2.2.14 in \cite{Stichbook09}] 
Let $C_{\cL}(D,G)$ be an AG-code with $D=P_1+\cdots+P_n$.
	\begin{enumerate}[$(a)$]
		\item If  $G' \sim G$  then $C_{\cL}(D,G') \sim C_{\cL}(D,G)$. \sk 
\item If $\mathcal{C}$ is a linear code over $\ff_q$ equivalent to $C_{\cL}(D,G)$, then there exists a divisor $G'\sim G$ such that $G'$ and $D$ have disjoint support and $\mathcal{C}=C_{\cL}(D,G')$.
		\item Moreover, for $z \in \cF$ such that $z(P_i) \neq 0$ for $i=1, \dots, n$ we have that $$C_{\mathcal{L}}(D,G-(z))= {\bf x} \cdot C_{\mathcal{L}}(D,G),$$ where ${\bf x}=(z(P_1), \dots , z(P_n)) \in (\fq^*)^n$. \sk 
	\end{enumerate}
\end{proposition}

\begin{proof}
See for instance Proposition 2.2.14 in \cite{Stichbook09}.
\end{proof}

The dual code $\cC^\perp$ of a linear  code $\cC$ over $\fq$  is the orthogonal complement of $\cC$ in $\fq^n$ with the standard inner product of $\fq^n$. It was shown in Corollary 2.6 of \cite{Stich1988} that the dual of a linear AG-code is also a linear AG-code formed with divisors that can be explicitly constructed. More precisely we have the following
\begin{proposition}%[Corollaries 2.7, 3.6 in \cite{Stich1988}, Propositions 2.2.10, 2.2.14 in \cite{Stichbook09}] 
\label{prop:AGdual}
The dual code of $C_{\cL}(D,G)$ is still a linear AG-code, in fact
\begin{equation} \label{AGdual}
	C_{\cL}(D,G)^\perp = C_{\cL}(D,D-G+(\eta)) 
\end{equation}
where $\eta$ is a Weil differential in $\cF$ such that 
$v_{P_i}(\eta)=-1$ and $\eta_{P_i}(1)=1$, for $i=1,\ldots, n$. 
\end{proposition}

%\begin{proof}
%	See for instance Proposition 2.2.10 in \cite{Stichbook09}. 
%\end{proof}

Recall that a linear code $\cC$ is said to be  {\em self-orthogonal} if $\cC \subset \cC^\perp$ and {\em self-dual} if $\cC=\cC^\perp$. The notion of equivalence of linear codes allows one to define the following class of linear codes generalizing self-dual codes.
\begin{definition}
Let ${\bf x} \in (\fq^*)^n$. A code $\cC \subset \fq^n$ is called \textit{$\bf x$-iso-dual} if $\cC^\perp=\bf x\cdot \cC$. We will simply speak of \textit{iso-dual} codes when there is no need to specify the vector ${\bf x}$.
\end{definition}

Clearly the class of iso-dual codes contains the class of the self-dual codes, that is, $\bf x$-iso-dual codes are just self-dual codes with  ${\bf x}=(1,1,\ldots,1)$. 
We observe that our definition of $\bf x$-iso-dual code is a particular case of iso-dual code where the equivalence of codes is given in terms of the monomial equivalence of codes, that is codes which are isometric to their duals via a permutation of coordinates and multiplying certain coordinates by non-zero constants (see, for instance, \cite{KL2017}).

We now recall a necessary condition for an AG-code to be iso-dual and a well known  characterization of iso-dual AG-codes constructed from  rational function fields. 
\begin{proposition} \label{prop:isodual}
Let $C_{\cL}(D,G)$ be an AG-code of length $n$  defined as in \eqref{defAGcode} over a function field $\cF$ of genus $g$. 
  \begin{enumerate}[$(a)$]
	\item $C_{\cL}(D,G)$ is iso-dual if there exists a canonical divisor $W$ such that $W\sim 2G-D$. More precisely, 
	$$C_{\cL}(D,G)^\perp={\bf x} \cdot C_{\cL}(D,G)$$ 
	where ${\bf x}=({\rm res}_{P_1}(\omega), \ldots, {\rm res}_{P_n}(\omega))$ with $\omega$ a Weil differential such that $2G-D=(w)$. \sk \label{isodualcond}

	\item If $\cF$ is a rational function field and $n\geq 2$ is even, then $C_{\cL}(D,G)$ is iso-dual if and only if $\deg G=\tfrac 12 (n-2)$. \sk \label{isorat}
	
	\item In particular, if $D$ and $G$ are divisors of a rational function field with disjoint support and $C_{\cL}(D,G)$ is an iso-dual AG-code, then $2G-D$ is a canonical divisor.\label{isoratcan}
  \end{enumerate}
\end{proposition}
\begin{proof}
	If $W$ is a canonical divisor equivalent to $2G-D$ then  $W'=2G-D$ is also a canonical divisor of $\cF$.
%so that  $W'=2G-D$ is a canonical divisor of $\cF$.
	 Since $G$ and $D=P_1+\cdots+P_n$ have disjoint support we have that $\nu_{P_i}(W')=-1$ for $i=1,\ldots,n$ and thus item \eqref{isodualcond} follows immediately from Supplement to Theorem 3.1 of \cite{Stich1988}. Item \eqref{isorat} appears as Corollary 3.6 of \cite{Stich1988}  and  item \eqref{isoratcan} is a direct consequence of item \eqref{isorat} and the characterization of canonical divisors given in Proposition  1.6.2 of \cite{Stichbook09}. 
\end{proof}

\begin{remark}\label{stichexample}
		Item \eqref{isoratcan} of  Proposition \ref{prop:isodual} does not hold in general. In Example 4.1.2 of \cite{Stich1988} a self-dual AG-code $C_\rr(D,G)$ is constructed such that $2G-D$ is not a canonical divisor. This was, in fact, a counterexample to a theorem in \cite{DM1985} where a characterization of a certain class of self-dual AG-codes defined over elliptic function fields was presented.
\end{remark}

\section{Iso-dual AG-codes over the Hermitian function field} \label{sec:isoHerm}
In general, it is a challenging problem to construct iso-dual codes in a function field since it strongly depends on the knowledge of differentials and 
divisors with suitable properties. We illustrate the problem in this section by presenting a family of iso-dual codes over the Hermitian function field that are self-dual in some cases. In the next section we will propose a method to lift a 
known iso-dual code in an extension of function fields. This method will provide an easier way to obtain iso-dual codes over an algebraic curve.

The Hermitian function field $\mathcal{H}/\fqs$, as a tamely ramified extension of $\ff_{q^2}(x)$, is the function field $\mathcal{H}=\fqs(x,y)$  given by the affine equation 
\begin{equation} \label{eq:Hermitian}
	y^{q+1}=x^q+x.
\end{equation}
It is well known (see Lemma 6.4.4 of \cite{Stichbook09}) that $\mathcal{H}$ has $q^3+1$ rational places and genus 
	$$g=\tfrac12 q(q-1).$$
For $i=1,\ldots,q$, let $Q_i$ be the totally ramified places; that is, $Q_i$ is a place over $P_i$ in the rational function field $\mathcal{F}=\fqs(x)$ corresponding to a root $\a_i$ of $x^q+x=0$. 
For each $\a \in \fqs$ with $\a^q+\a \ne 0$ we have $q+1$ places in $\fqs(x,y)$ over the place $P_\a$ in $\fqs(x)$ associated to the zero of the function $x-\a$. 

In 2021, L.\@ Sok proposed many families of self-dual codes $C_\cL(D,G)$ over the Hermitian function field in the case the divisor $G$ has support in one point (the only pole of $x$) in the function field (see \cite[Theorem 9]{S2021}).
Here we propose  families of iso-dual codes over the Hermitian function field where the support of $G$ consist of all the totally ramified places in $\mathcal{H}/\fqs(x)$.

\begin{theorem} \label{thm:herm}
Let $\mathcal{H}=\fqs(x,y)$ be the Hermitian function field of genus $g=\tfrac 12 q(q-1)$
defined by $y^{q+1}=x^q+x$. 
Let $\beta$ a non-zero integer and consider the disjoint divisors
	$$D=\sum_{ \a^q +\a \ne 0} Q | P_\a , 
	\qquad \text{and} \qquad G = 
	\big( \tfrac12 (q^3+q^2-2q-2)
	 -q\b \big) Q_\infty +\beta \sum_{i=1}^q Q_i,$$ 
	where the $Q_i$ are all the totally ramified places in $\mathcal{H}/\fqs$ associated to a root $\a_i$ of $x^q+x=0$ and $\deg G >0$.
Then, the code $C_\cL(D,G)$ is ${\bf x}$-iso-dual with parameters 
	$$[n=q^3-q, \:\, k=\tfrac 12 (q^3-q), \:\, d \ge \tfrac 12 q^2(q-1)+1],$$
where ${\bf x}=(z(P_1),\ldots,z(P_n))$ with $z=y^{2\beta+2-q^2}$.
Moreover, if $q^2-1$
is a divisor of $2\b+1$ (hence $q$ is even), then $C_\cL(D, G)$ is self-dual. 
\end{theorem}

\begin{proof}
We will use $(a)$ in Proposition \ref{prop:isodual}. Let $\mathcal{H}= \fqs(x,y)$ be the Hermitian function field defined by the affine equation $y^{q+1}=x^q+x$ and $\mathcal{F}=\fqs(x)$ the rational function field.
Consider the function $t \in \ff_p(x)$, where $p=\textrm{Char}(\ff_q)$, 
given by
\begin{equation} \label{eq:t}
	t=\frac{x^{q^2}-x}{x^q+x}.
\end{equation}	
Then, we clearly have that $v_{P}(t)=1$ for any $P \in \supp(D)$. Thus, from Proposition~8.1.2 in \cite{Stichbook09}, the Weil differential $\eta:=\frac 1t dt$ satisfies $v_{P}(\eta)=-1$ and ${\rm res}_{P}(\eta)=1$ for any $P \in \supp(D)$. Hence, if $W:=(\frac 1t dt)^\cH$ is the canonical divisor of the differential $\frac 1t dt$ in $\cH$, we have that 
	$$C_\cL(D, G)^\perp = C_\cL(D, D-G+W).$$

Now, we explicitly compute $W=\left(\frac 1t dt\right)^\cH = \left(\frac 1t \frac{dt}{dx}dx\right)^\cH$. 
From \eqref{eq:t} 
we have 
	$$dt = -\frac{x^{q^2}+x^q}{(x^q+x)^2} dx = -(x^q+x)^{q-2} dx.$$ 
Hence, from $\diff(\cH/\cF) = \sum_{i=1}^q q Q_i + q Q_\infty$, we compute the divisor $W$
\begin{align*}
W &= \Big(-(x^q+x)^{q-2} \, \frac{x^q+x}{x^{q^2}-x} \, dx \Big)^\cH  \\
  &= \Big( \frac{(x^q+x)^{q-1}}{x^{q^2}-x} \Big)^\cH + (dx)^\cH  \\
  &= (q-1)(x^q+x)^\cH-(x^{q^2}-x)^\cH-2(x)_\infty^\cH+	\diff(\cH/\cF)  \\
  &= (q^2-1) \Big(\sum_{i=1}^q Q_i-qQ_\infty \Big) - \Big(D+\sum_{i=1}^q (q+1)Q_i-(q^3+q^2)Q_\infty \Big) \\ 
  & \qquad - 2(q+1)Q_\infty+q\sum_{i=1}^q Q_i + q Q_\infty  \\
  &=-D+(q^2-2)\sum_{i=1}^q  Q_i+(q^2-2)Q_\infty.
\end{align*}
Now, we notice that $2G-D-W$ equals  
$$(q^3+q^2-2q-2 -2q\b ) Q_\infty + 2\beta \sum_{i=1}^q Q_i-D- \big(-D+(q^2-2) \sum_{i=1}^q  Q_i+(q^2-2)Q_\infty \big)$$ 
and, hence, we have that
\begin{align*}
2G-D-W %&=(q^3+q^2-2q-2 -2q\b ) Q_\infty +2\beta \sum_{i=1}^q Q_i-D-(-D+(q^2-2) \sum_{i=1}^q  Q_i+(q^2-2)Q_\infty)\\
&=(q^3-2q -2q\b ) Q_\infty+(2\b-q^2+2)  \sum_{i=1}^q Q_i =(y^{2\beta +2-q^2})^\cH.
\end{align*}
Let	$z:=y^{2\beta +2-q^2} $. Then, we have obtained that $D-G+W=G-(z)^\cH$ and we can show that the code is iso-dual. In fact, it is ${\bf x}$-iso-dual since
\begin{align*}
C_\cL(D, G)^\perp = C_\cL(D, D-G+W) 
				  = C_\cL(D, G-(z)^\cH ) 
				  = {\bf x} \cdot C_\cL(D, G),
\end{align*} 
where $ {\bf x}=(z(P_1), \ldots, z(P_n)) \in (\fqs^*)^n$ is the vector given by the computation of the function $z$ in all the places $P$ in $\supp(D)$.

The length of the code is given by 
$n=\deg(D)=q^3-q$. The dimension  is $k=\tfrac 12 n$ since the code is iso-dual. 
By \eqref{dimcode}, the minimum distance satisfies $d\ge n-\deg (G)=\frac{q^3-q^2+2}{2}$.

Finally, $C_\cL(D, G)$ is self-dual if and only if the vector ${\bf x}$ equals $(1,\ldots,1)$. Let 
	$$x_i=z(P_i)= y(P_i)^{2\b +2-q^2}\in \fqs.$$ 
Then, 
$C_\cL(D, G)$ is 
self-dual if and only if 
	$$x_i^{q^2-1}= x_i^{2\b +2-q^2 }=1$$ 
for every $i=1,\ldots,n$.
Choosing $\b$ such that $q^2-1$ divides $2\b +1$ 
we obtain that the code $C_\cL(D, G)$ is self-dual.
\end{proof}

For example, for any choice of $\beta \in \Z\smallsetminus \{0\}$ 
we get Hermitian iso-dual AG-codes (and self-dual AG-codes if we further take 
$\beta=\tfrac 12((q^2-1)t-1)$ for some $t\in \N$ ) 
with parameters $[60,30,\ge \! 25]$ over  $\mathbb{F}_{4^2}$, $[720, 360, \ge \! 325]$ over $\mathbb{F}_{9^2}$, $[4080, 2040, \ge 1921]$ over $\mathbb{F}_{16^2}$ and $[15600, 7800, \ge 7501]$ over $\mathbb{F}_{25^2}$.

\section{Lifting iso-dual AG-codes on function field extensions} \label{sec:lifting}
Let $\cM/\cF$ be a finite and separable extension of function fields defined over $\fq$. Here we provide a construction that allows to lift an iso-dual dual AG-code defined over $\cF$ to an iso-dual AG-code defined over $\cM$. 

Given an iso-dual AG-code $C_{\cL}(D,G)$ over $\cF$,  we want to lift it to an iso-dual code over $\cM$, that is,  we want to define a code
$\tilde C_\cL =C_\cL(\tilde D,\tilde G)$ over $\cM$ such that for some ${\bf \tilde x} \in \fq^n$ we have
	$$ 
	C_{\cL}(\tilde D, \tilde G)^\perp = \tilde {\bx} \cdot C_{\cL}(\tilde D, \tilde G).$$

In the following result, we give conditions on the function field extension $\cM / \cF$ for an iso-dual code defined over $\cF$ to be lifted to an iso-dual code defined over $\cM$.

\begin{theorem} \label{teolevantado}
	Let $\cM/\cF$ be a finite separable extension of function fields over $\fq$ of degree $m\geq 2$ with genera $g_{_\cM} $ and $g_{_\cF}$, respectively. Let $n\geq 1$ be an even integer and suppose that $\{P_1, \dots, P_n\}$ and $\{Q_1, \dots, Q_r\}$ are disjoint set of places of $\cF$ such that:
		\begin{enumerate}[$(a)$]
			\item $P_1, \dots, P_{n}$ are rational places and  $P_i$  splits completely in $\cM/\cF$ for $1\leq i\leq n$, \sk
			
			\item  the extension $\cM/\cF$ is unramified outside the set $\{Q_1, \dots , Q_{r}\}$ for some $r \geq 1$, and \sk 
			
			\item for each $1\leq i\leq r$ and each place $R$ of $\cM$ lying over $Q_i$ the different exponent $d(R|Q_i)$ is even.
		\end{enumerate} 
Let $(\b_1, \ldots, \b_r) \in \Z^r$ be a non zero $r$-tuple and consider the divisors of $\cF$
		$$D= \sum_{i=1}^n P_i  \qquad \text{and} \qquad G=\sum_{i=1 }^r \b_i Q_i.$$
If the divisor $2G-D$ is equivalent to a canonical divisor, then  the AG-code $C_{\cL}(D,G)$ defined over $\cF$ is iso-dual and the AG-code $C_{\cL}(\tilde D, \tilde G)$ defined over $\cM$ where
	$$\tilde D = \con_{\cM/\cF}(D) \qquad \text{and} \qquad \tilde G = \con_{\cM/\cF}(G) + \tfrac 12 \mathrm{Diff}(\cM/\cF)$$
is  also iso-dual.
\end{theorem} 

\begin{proof}
From \eqref{isodualcond} of Proposition \ref{prop:isodual} we see that the AG-code $C_\rr(D,G)$ is iso-dual. Let us now consider the divisor $\tilde D$ of $\cM$ defined as the conorm of $D$ in $\cM/\cF$, that is 
		$$\tilde D = \text{Con}_{\cM/\cF}(D) = \sum_{i=1}^n\sum_{j=1}^m R_{i,j},$$
where $R_{i,1},\ldots,R_{i,m}$ are all the places of $\cM$ lying over $P_i$ for each $i=1,\ldots,n$. 

By hypothesis we have that there exist a Weil differential $\eta$ of $\cF$ and an element $f\in \cF$ such that	
	$$(\eta)=2G-D+(f)^F.$$ 
Let $\tilde{\eta}=\Cot (\eta)$ be the cotrace of $\eta$ in $M$.  
Then, by \eqref{eq:cotrace}, 
we have that
\begin{align*}
	(\tilde{\eta})  &=\con_{\cM/\cF}((\eta)) + \mathrm{Diff}(\cM/\cF)  \\
	&= \con_{\cM/\cF}(2G-D+(f)^F) + \mathrm{Diff}(\cM/\cF)  \\
 &= -\tilde{D} + \con_{\cM/\cF}(2G+(f)^F) + \mathrm{Diff}(\cM/\cF).   
\end{align*}

Consider the divisor $A=2G+(f)^F$, hence $(\eta)=A-D$. 
There are two possibilities: either $\supp(D)\cap\supp(A)=\varnothing$ or not.

In the first case, we clearly have that $v_{\tilde R}(\tilde \eta)=-1$ for all ${\tilde R}\in \supp(\tilde{D})$. 
Since $ \mathrm{Diff}(\cM/\cF)=\sum_{i=1}^r\sum_{S|Q_i}d(S|Q_i)S$ and, by hypothesis, $d(S|Q_i)$ is even for every $i=1,\ldots, n$, we have a well defined divisor 
	\[\tilde G = \con_{\cM/\cF}(G) + \tfrac 12 \mathrm{Diff}(\cM/\cF),\] 
of $\cM$. Since $\tilde{D}$ and $\tilde{G}$ have disjoint supports and $v_{\tilde R}(\tilde \eta)=-1$ for all ${\tilde R}\in \supp(\tilde{D})$, 
from Proposition \ref{prop:isodual}
we have that
	\begin{equation}\label{xdual1}
			C_{\cL}(\tilde D, \tilde G)^\perp =\mathbf{a}\cdot C_{\cL}( \tilde D, \tilde D- \tilde G + (\tilde{\eta})),
	\end{equation}
where $\mathbf{a}=(\tilde{\eta}_{R_{1,1}}(1),\ldots,\tilde{\eta }_{R_{n,m}}(1))\in \fq^{nm}$. On the other hand, by the definition of %the divisor 
$\tilde{G}$, we also have that
	\begin{align*} 
		(\tilde{\eta}) & = \con(2G-D+(f)^\cF)+ \mathrm{Diff}(\cM/\cF) \\
                       & = 2\con_{\cM/\cF}(G)-\tilde{D} + \con_{\cM/\cF}((f)^\cF) +  \mathrm{Diff}(\cM/\cF) \\
					   & = 2 \tilde G - \tilde D + (f)^\cM,
	\end{align*}
where $(f)^\cM=\con_{\cM/\cF}((f)^\cF)$ by Proposition 3.1.9 of \cite{Stichbook09}. Thus $\tilde D-\tilde G+ (\tilde{\eta})=\tilde G + (f)^\cM$ and this means that $\tilde D-\tilde G+ (\tilde{\eta})\sim\tilde G$ so that
	\begin{equation}\label{xdual2}
		C_{\cL}( \tilde D, \tilde D- \tilde G + (\tilde{\eta}))=\mathbf{b} \cdot C_{\cL}(\tilde D, \tilde G),
	\end{equation}
for some vector ${\bf b}=(b_1,\ldots,b_{nm}) \in \fq^{nm}$. If
	\[\mathbf{x}=(\tilde{\eta}_{R_{1,1}}(1)b_1,\ldots,\tilde{\eta }_{R_{n,m}}(1)b_{nm})\in \fq^{nm},\]
from \eqref{xdual1} and \eqref{xdual2} we deduce that 
	\[C_{\cL}(\tilde D, \tilde G)^\perp=\mathbf{x}\cdot C_{\cL}(\tilde D, \tilde G),\]
that  is  $C_{\cL}(\tilde D, \tilde G)$ is an iso-dual AG-code over $\cM$ which is a lifting to $\cM$ of the iso-dual AG-code $C_{\cL}(D,G)$ over $\cF$.
		
If the second case occurs, that is if $\supp(D)\cap\supp(A)\neq\varnothing$, we know that there exists a divisor $A'\sim A$ such that $\supp(D)\cap\supp(A')=\varnothing$. Thus, there exists an element  $h\in \cF$ such that $2G+(f)^\cF=A=A'+(h)^\cF$ and then 
		\[
			(h^{-1}\eta) = (\eta)-(h) = -D+2G+(f)^\cF-(h)^\cF = -D+2G+(\tfrac fh)^\cF = -D+A',
		\]   
		with $\supp(D)\cap\supp(A')=\varnothing$. Now, $\omega=h^{-1}\eta$ is a Weil differential of $\cF$ and thus, if we define $\tilde{\omega}= \Cot(\omega)$, we have
		\begin{align*}
			(\tilde{\omega})  &=\con_{\cM/\cF}((\omega)) + \mathrm{Diff}(\cM/\cF)  \\
			&= \con_{\cM/\cF}(-D+A') + \mathrm{Diff}(\cM/\cF)  \\
			&= -\tilde{D} + \con_{\cM/\cF}(A') + \mathrm{Diff}(\cM/\cF),  
		\end{align*}
	where  $\supp(D)\cap\supp(A')=\varnothing$. Therefore, we are in the conditions of the first case just considered and we see that the iso-dual AG-code $C_{\cL}(D,G)$ defined over $\cF$  can also be lifted to an iso-dual AG-code defined over $\cM$.
\end{proof}

 \begin{remark}\label{canonicallifting}
	From the proof of Theorem \ref{teolevantado} we see that starting with two divisors $D=P_1+\cdots+P_n$ and $G$ of the function field $\cF$ with disjoint support such that the divisor $2G-D$ is equivalent to a canonical divisor, not only the lifted code $\mathfrak{L}_{\cM/\cF}(C_\cL(D,G))=C_\rr(\tilde{D},\tilde{G})$ is  iso-dual but also the divisor $2\tilde{G}-\tilde{D}$ is a canonical divisor of the function field $\cM$.  In view of Remark \ref{stichexample} this is an interesting property of our construction of iso-dual AG-codes given in Theorem \ref{teolevantado}. 
\end{remark}
 
\begin{remark}
Notice that in the above theorem we start with two divisors $D$ and $G$ of $\cF$ such that the divisor $2G-D$ is equivalent to a canonical divisor. Even if we have that the AG-code $C_\rr(D,G)$ is iso-dual we can not ensure that the divisor $2G-D$ is equivalent to a canonical divisor of $\cF$, as we have seen in Remark \ref{stichexample}. However if $\cF$ is a rational function field and $\deg(G)=(n-2)/2$, then the divisor $2G-D$ is a canonical divisor by Proposition \ref{prop:isodual}. There are other instances (see, for example, \cite{DM1985} and the comment made in Example 4.1.2 of \cite{Stich1988}) where having an iso-dual AG-code $C_\rr(D,G)$ is equivalent to having the divisor $2G-D$ being equivalent to a canonial divisor, so that we can apply our Theorem 4.1 in these cases. 

Notice also that in the above theorem we use extensions of function fields where the different exponents are even integers. This is not a huge constraint because the condition on the evenness of the different exponents holds in several well-known situations like: 
\begin{enumerate}[($i$)]
	\item the ones defined by Artin-Schreier extensions in odd characteristic, \sk 
	
	\item tamely ramified extensions of odd degree, \sk 
	
	\item the case of the so called weakly ramified extensions (see Definition~7.4.12 of \cite{Stichbook09}), and also \sk
	
	\item the case where there is a rational place $P$ of $\cF$ totally ramified in $\cM$ and the extension $\cM/\cF$ is unramified outside the place $P$. This is so because in this case if $Q$ is the only place of $\cM$ lying over $P$, then $Q$ is also a rational place and then by Hurwitz genus formula we have that
	\begin{equation}\label{oneramified}
		2g_{_\cM}-2=(2g_{_\cF}-2)[\cM:\cF]+d(Q|P),
	\end{equation}
	showing that $d(Q|P)$ is even. 
\end{enumerate}
\end{remark}

It is also worth noticing that in the presence of more than one place of $\cF$ ramified in $\cM$, either a one-point or a multi-point iso-dual code $C_\cL(D,G)$ over $\cF$ %that 
can be lifted to an iso-dual code over $\cM$, according to the choice of the number of zero coordinates in the $r$-tuple $(\beta_1,\ldots,\beta_r)$ in the above theorem. However, despite having some flexibility in choosing the $r$-tuple  $(\beta_1,\ldots,\beta_r)$ defining the divisor $G$, these integers must satisfy the equation
	\[ \sum_{i=1}^r\beta_i\deg(Q_i) = \tfrac{1}{2}(n+2g_{_\cF}-2),\]  
according to ($b$)
%item \eqref{isoeta} 
of Proposition \ref{prop:isodual}, because we are requiring that $C_\cL(D,G)$ is an iso-dual code.

In view of the result obtained in Theorem \ref{teolevantado},  
we make the following definition and notation.

\begin{definition} \label{lifted}
Given a function field extension $\cM/\cF$ and an AG-code $C_{\cL}(D,G)$ over $\cF$ as in Theorem \ref{teolevantado}, the AG-code $C_{\cL}(\tilde D, \tilde G)$, where 
	$$\tilde D = \text{Con}_{\cM/\cF}(D) \qquad \text{and} \qquad \tilde G = \con_{\cM/\cF}(G) +\tfrac{1}{2} \mathrm{Diff}(\cM/\cF),$$ 
will be called the {\it lifted code} (or \textit{the lift}) of $C_{\cL}(D,G)$ 
to $\cM$, and we will denote it by $\mathfrak{L}_{\cM/\cF}(C_\cL(D,G))$. 
\end{definition}

With Definition \ref{lifted}, Theorem \ref{teolevantado} says that if $C_{\cL}(D,G)$ is an iso-dual AG-code over $\cF$ with $2G-D$ equivalent to a canonical divisor of $\cF$, then the lifted code $\mathfrak{L}_{\cM/\cF}(C_\cL(D,G))$ defined over $\cM$ is also  iso-dual.

We now give an estimate of the parameters of the lifted code of an iso-dual code.
\begin{corollary} \label{parameters}
Under the same conditions of Theorem \ref{teolevantado}, if $mn> 2g_{\mathcal{M}}-2$ 
then $\mathfrak{L}_{\cM/\cF}(C_\cL(D,G))$ is  an $[\tilde{n}, \tilde{k}, \tilde{d}]$-code where  
	$$\tilde{n}=mn, \qquad \tilde{k} = \tfrac12 mn, \qquad \text{and} \qquad \tilde{d} \ge \tfrac 12 (mn-2g_{\mathcal{M}}+2).$$
\end{corollary}
\begin{proof}
Clearly, $\tilde{n}=mn$ because the support of $\tilde{D}$ consists of exactly $mn$ (rational) places of $\cM$, by definition of $\tilde{D}$. Also, since $C_{\cL}(\tilde D, \tilde G)$ is an iso-dual code, it holds that $\tilde{k}=\frac12 {mn}$ and from Proposition \ref{prop:isodual} we have that
	\[\deg(\tilde{G}) = \tfrac 12 (mn+2g_{\mathcal{M}}-2). \] 	
From Proposition \ref{prop:dimcode} 
\[\tilde{d} \ge \tilde n -\deg \tilde{G} = mn-\tfrac 12(mn+2g_{\mathcal{M}}-2) = \tfrac 12 (mn-2g_{\mathcal{M}}+2),\]
as we wanted to show.
\end{proof}

\section{Lifting iso-dual codes from the rational function field} \label{sec: from rational}
We use now Theorem \ref{teolevantado} to construct iso-dual AG-codes in function fields by lifting rational iso-dual codes, that is iso-dual codes over the rational function field $\fq(x)$. This is the simplest and most favorable situation to apply our method because of items \eqref{isorat} and \eqref{isoratcan} of Proposition~\ref{prop:isodual}.

In $\fq(x)$ we choose pairwise distinct rational places $P_1, \dots, P_n$ for $n \geq 4$ even and we denote by $P_\infty$  the only pole of $x$ in $\fq(x)$. By considering the divisors 
\[D=P_1+\cdots+P_n \qquad \text{and} \qquad G=\tfrac12 (n-2) P_\infty,\]
of $\fq(x)$ we have not only that the AG-code $C_\cL(D,G)$ is an iso-dual code over $\fq(x)$ but also that the divisor $2G-D$ is a canonical divisor according to items \eqref{isorat} and \eqref{isoratcan} 
of Proposition~\ref{prop:isodual}.
Now, %In the rest of the section, 
we are going to choose suitable places $P_1, \dots, P_n$ in $\fq(x)$ in order to be able to lift the iso-dual $C_\cL(D,G)$ over $\fq(x)$ to some elementary abelian $p$-extensions, and also to some Kummer extensions in the next section.

Some of the examples given in this section include the important class of {\it maximal function fields}, that is, a function field $\cF$ over $\fqs$ of genus $g(\cF)$ such that its number $\cF(\fqs)$ of $\fqs$-rational points attains the upper bound on the Hasse-Weil bound, that is
$$ \cF(\fqs)=q^2+1+2g(\cF)q.$$
\begin{remark}
A distinguished  example of maximal function field is the Hermitian function field which was already considered in Section \ref{sec:isoHerm}  as a tamely ramified extension of $\ff_{q^2}(x)$. In that case the constructed iso-dual code was a multi-point AG-code but here, with the Hermitian function field considered as an elementary abelian $p$-extension of $\ff_{q^2}(x)$, the constructed iso-dual code  will be a one-point AG-code. An important advantage we have with the elementary abelian $p$-extensions of $\fq(x)$, is that it is  rather easy to find a generating matrix for the lifted iso-dual codes in many cases. 
\end{remark}

\subsection*{Elementary abelian $p$-extensions}
Let $f$ be a polynomial over $\ff_{q^s}$, $s\ge 1$, of degree $m>0$ such that $(m, q)=1$ and let $0\neq \mu\in \ff_{q^s}$. Suppose the polynomial $t^q+\mu t\in\ff_{q^s}[t]$
splits completely into linear factors over $K=\ff_{q^s}$. A function field of the form $\cF=K(x,y)$ where 
\begin{equation} \label{eq:p-abeliana}
	y^q+\mu y=f(x),
\end{equation}
is called an {\em elementary abelian $p$-extension}  of  $K(x)$. 
From Proposition 6.4.1 of \cite{Stichbook09} we have that the extension $\cF/K(x)$ is of degree $q$ and  genus 
$$g = \tfrac 12 (q-1)(m-1).$$ 
The only pole $P_\infty$ of $x$ in $K(x)$ is totally ramified in $\cF$ and the extension $\cF/K(x)$ is unramified outside the place $P_\infty$. Since $P_\infty$ is a rational place we have from \eqref{oneramified} that the different exponent
\begin{equation}\label{expdif}
d(Q_\infty|P_\infty)=(q-1)(m+1),
\end{equation}
is even, where $Q_\infty$ is the only place of $\cF$ lying over $P_\infty$.

We now give sufficient conditions to get lifted iso-dual codes over a general elementary abelian $p$-extension.

\begin{proposition}
For any $s\in \N$ let us consider the elementary abelian $p$-extension 
$\cF=\ff_{q^s}(x,y)$ as in \eqref{eq:p-abeliana}. Let $n \geq 4$ be an even integer and suppose that $P_1,\ldots, P_n$ are $n$ different rational places of  $\ff_{q^s}(x)$ such that each $P_i$ splits completely in $\cF / K(x)$. Then there exists a lifted iso-dual AG-code over $\cF$ with parameters
	$$[nm, \, \tfrac 12 nm, \, \ge \tfrac 12 (m(n-q+1)+q-3)],$$	
and generating matrix 
\begin{equation} \label{genmatrix}
	M = \big( \alpha_i^a \beta_{i,q}^b : 0\le a, 0 \le b<q, qa+mb<r \big)_{1\le i\le n},
\end{equation}
with $r=\tfrac 12(q(n+m-1)-m-1)$, where $\alpha_i$ is such that $P_i$ is the only  zero of $x-\alpha_i$ in $\ff_{q^s}(x)$, provided that  $\beta_{i,q},\ldots,\beta_{i,q}$ are the $q$ different roots of $T^q+  \mu T =f(\alpha_i)$ for $1\le i \le n$. 
\end{proposition}

\begin{proof}
By hypothesis we have that $P_1,\ldots, P_n$ are $n$ different rational places of  $\ff_{q^s}(x)$ with $n\geq 4$ even and such that each $P_i$ splits completely in $\cF / \ff_{q^s}(x)$. By taking the divisors
	\[D=P_1+\cdots+P_n \qquad \text{and} \qquad G=\tfrac12 (n-2) P_\infty,\]
we have that the AG-code $C_\cL(D,G)$ is an iso-dual code over $\ff_{q^s}(x)$ according to ($c$) of Proposition \ref{prop:isodual}.  We are now in the conditions of Theorem~\ref{teolevantado} and then we have that $\mathfrak{L}_{\cF/\ff_{q^s}(x)}(C_\cL(D,G))$ is an iso-dual code over $\cF$. In fact we have that 
	$$ \mathfrak{L}_{\cF/\ff_{q^s}(x)}(C_\cL(D,G))=C_\cL(\tilde{D},\tilde{G})=C_\cL(\tilde{D},rQ_\infty), $$
where $\tilde{D}=\sum_{i=1}^n\sum_{j=1}^qR_{i,j}|P_i$ and $R_{i,1},\ldots,R_{i,q}$ are the all the places of $\cF$ lying over $P_i$ for $i=1,\ldots,n$ and 
	\[ r=\tfrac 12(q(n+m-1)-m-1).\] 
The latter  is because 
	\[ \textrm{Diff}(\cF/K(x))=d(Q_\infty|P_\infty)Q_\infty=(2q+(q-1)(m-1)-2) Q_\infty,\]
and so
\begin{align*}
	\tilde{G} &= \con_{\cF/\ff_{q^s}(x)}(\tfrac 12 (n-2) \, P_\infty) + \tfrac 12 \textrm{Diff}(\cF/\ff_{q^s}(x))\\
	&= \tfrac 12 q(n-2) \, Q_\infty + \tfrac 12 (2q+(q-1)(m-1)-2) \, Q_\infty\\
	&= \tfrac 12 (q(n+m-1)-m-1) \, Q_\infty.
\end{align*}

The parameters of the lifted code follow from Corollary \ref{parameters} and the expression of the genus of $\cF$.

On the other hand, from Proposition 6.4.1 of \cite{Stichbook09} we know that the set 
\begin{equation} \label{basisRR}
	\{x^ay^b: 0\leq a, 0\leq b< q, qa+mb\leq r\}, 
\end{equation}
is an $\ff_{q^s}$-basis of $\cL(rQ_\infty)$. Now for each $1\leq i\leq n$ the place $P_i$ is the zero of $x-\alpha_i$ in $\ff_{q^s}(x)$ and we are assuming that the polynomial $T^q+\mu T=f(\alpha_i)$ has $q$ different roots $\beta_{i,1},\ldots,\beta_{i,q}\in \ff_{q^s}$. Then by Kummer theorem both elements $x-\alpha_i$ and $y-\beta_{i,j}$ belong to $ R_{i,j}$ for $j=1,\ldots,q$ and then the residual class $(x^ay^b)(R_{i,j})$  is
\begin{equation} \label{eq: genmatrix}
	(x^ay^b)(R_{i,j}) = \alpha_i^a\beta_{i,q}^b.
\end{equation}
This gives us the explicit generator matrix $M$ of  $\mathfrak{L}_{\cF/\ff_{q^s}(x)}(C_\cL(D,G))$ as stated in \eqref{genmatrix}.
\end{proof}

The following concrete example shows that the estimates given in Corollary \ref{parameters} can not be improved in general.
\begin{example} \label{ex:844}
	Consider the polynomial $t^2+ t\in \ff_8[t]$ and let $\ff_8=\ff_2(\alpha)$ where $\alpha$ satisfies $\alpha^3+\alpha+1=0$. 
	The equation
	\[ y^2+y = x^3 \]
	defines an elementary abelian $2$-extension $\cF$ of $\ff_8(x)$ of genus $g=1$. It is easy to check using \textsc{Sage} \cite{sage} that the polynomial $T^2+T=f(\gamma)$ splits into $2$ different linear factors over $\ff_8$ if $\gamma\in\{0,\alpha+1,\alpha^2+1,\alpha^2+\alpha+1\}$. Therefore, we have a rational iso-dual AG-code of the form  					
	$$C_\cL(P_1+P_2+P_3+P_4,P_\infty)$$  
	that can be lifted to an iso-dual AG-code over $\cF_{\mu}$ where $P_1$ is the zero of $x$ in $\ff_8(x)$, $P_2$ is the zero of $x+\alpha+1$ in $\ff_8(x)$, $P_3$ is the zero of $x+\alpha^2+\alpha+1$ in $\ff_8(x)$ and $P_4$ is the zero of $x+\alpha^2+1$ in $\ff_8(x)$. By Corollary \ref{parameters} the lifted code is an AG-code of length $8$,  dimension $\tilde k= 4$ and minimum distance $\tilde d \geq 4$.
	
	We  show now that, in fact, $\tilde{d}=4$.  Since  in this case $q=2$, $m=3$ and $n=4$ we see that $r=4$ so that $\tilde{G}=4Q_\infty$ and we also have from \eqref{basisRR} that the set $\{1,x,x^2, y\}$ is an $\ff_8$-basis of $\cL(4Q_\infty)$. Using \eqref{genmatrix}  we form the matrix
	\[\left( \begin{array}{cccccccc}
		1 & 1 & 1 & 1 & 1 & 1 & 1 & 1\\
		0 & 0 & \alpha + 1 & \alpha + 1 & \alpha^{2} + \alpha + 1 & \alpha^{2} + \alpha + 1 & \alpha^{2} + 1 & \alpha^{2} + 1 \\
		0 & 0 & \alpha^{2} + 1 & \alpha^{2} + 1 & \alpha + 1 & \alpha + 1 & \alpha^{2} + \alpha + 1 & \alpha^{2} + \alpha + 1 \\
		0 & 1 & \alpha^{2} + \alpha & \alpha^{2} + \alpha + 1 & \alpha^{2} & \alpha^{2} + 1 & \alpha & \alpha + 1
	\end{array}  \right)\]
	which is a generator matrix of 
	$$C_\cL(\tilde{D},4Q_\infty) = \mathfrak{L}_{\cF_1/\ff_8(x)}(C_\cL(P_1+P_2+P_3+P_4,P_\infty)).$$ 
	With this matrix we can use now the Coding theory package of SAGE and see that $\tilde d=4$, as claimed. 
	
\end{example}

\subsection{Iso-dual codes over function fields covered by the Hermitian}
In the case of finite fields of quadratic cardinality there are  examples  of maximal function fields  which are elementary abelian $p$-extensions of $\ff_{q^2}(x)$ for any prime $p$. An  instance of this situation is the  function field $\cF$  defined by the equation
	\[ y^q+y=x^\ell, \]
where $\ell >1$ is a divisor of $q+1$, of genus 
	$$g=\tfrac 12 (q-1)(\ell-1)$$ 
(the case $\ell=q+1$ is the well known Hermitian function field over $\ff_{q^2}$). By Example~6.4.2 of \cite{Stichbook09} we have $1+(q-1)\ell$ rational places in $\ff_{q^2}(x)$  that splits in $\cF / \ff_{q^2}(x)$. Let $q \geq 3$ be odd and take $n=(q-1)\ell$. Then $n \ge 4$  is even and 
	\[qn = q(q-1)\ell > (q-1)(\ell-1)-2 = 2g-2.\]
The code $C_\cL(P_1+\cdots+P_n, \tfrac 12 (n-2) P_\infty)$ is a rational iso-dual code, by Proposition \ref{prop:isodual}, and the lifted code
	$$\mathfrak{L}_{\cF/\fq(x)}(C_\cL(P_1+\cdots+P_n, \tfrac 12 (n-2) P_\infty))$$ 
is an iso-dual AG-code over $\cF$ of length, dimension and minimum distance given by 
 	$$ \tilde{n}=q(q-1)\ell, \qquad \tilde{k} = \tfrac 12 q(q-1)\ell,  \qquad \text{and} \qquad 
 	\tilde{d} \ge \tfrac 12 ((q-1)^2\ell+q+1),$$
respectively.

\subsection{Iso-dual codes on the Suzuki curve}
Let $q=2^{2m+1}$ and $q_0=2^m, m \geq 1$. The Suzuki function field $\mathcal{S}_q=\mathbb{F}_{q^4}(x,y)$ is defined by the affine equation 
	\begin{equation} \label{eq: Sz}
		\mathcal{S}_q: \quad y^q + y = x^{q_0} (x^q + x).
	\end{equation}
This curve has genus 
	$$g(\mathcal{S}_q)=q_0(q-1)$$ 
and it is $\mathbb{F}_{q^4}$-maximal; with only one place at infinity, $P_\infty$, the only pole of $x$ and $y$.

We now show that we can lift iso-dual AG-codes over this function field.

\begin{proposition} \label{prop: Suzuki}
For $q =2^{2m+1}$ and $q_0=2^m$, with $m \geq 1$, there exist lifted iso-dual AG-codes on the curve  $\mathcal{S}_q$ over $\mathbb{F}_{q^4}$ 
with parameters satisfying
	\begin{equation} \label{eq: parameters Sz}
		\big[ q(q^3-q), \;  \tfrac 12 q(q^3-q), \; \ge \tfrac 12 (q^4-q^2-2q_0(q-1)+2) \big].
	\end{equation}
\end{proposition}

\begin{proof}
Let $f(x)= x^{q_0} (x^q + x)$. 
We notice that if $\a$ in $\mathbb{F}_{q^4}$ is such that $\mathrm{Tr}_{\mathbb{F}_{q^4} / \fq}(\a)=0$ and $a^q+ a \neq 0$ then 
$\mathrm{Tr}_{\mathbb{F}_{q^4} / \fq}(f(\a))=0$ and we conclude, from the 90's Hilbert Theorem, that there exists $q$ solutions $y \in \mathbb{F}_{q^4}$ to $y^q+y=f(\a)$. That is, the rational place $P_{x-a}$ in $\mathbb{F}_{q^4}(x)$ splits in the extension $\mathcal{S}_q / \mathbb{F}_{q^4}(x)$. 

Let $n=q^3-q$ and let $P_1, \dots, P_n$ be the $\mathbb{F}_{q^4}$-rational places splitting in the extension $\mathcal{S}_q / \mathbb{F}_{q^4}(x)$. The code $C_\cL(D,G)$ with 
	$$D=P_1+\cdots+P_n \qquad \text{ and } \qquad G=\tfrac 12 (n-2) P_\infty$$ 
is a rational iso-dual code by Proposition \ref{prop:isodual}, since $n$ is even. Thus, we are in the conditions of Theorem~\ref{teolevantado} and then $\mathfrak{L}_{\mathcal{S}_q/\mathbb{F}_{q^4}(x)}(C_\cL(D,G))$ is an iso-dual code over $\mathcal{S}_q$ with parameters as stated in \eqref{eq: parameters Sz}, by Corollary \ref{parameters}.
\end{proof}

\subsection{Iso-dual codes on an maximal curve covered by the $GGS$-curve}
From a result by J.\@ P.\@ Serre it is known that any $\fqs$-rational curve which is $\fqs$-covered by an $\fqs$-maximal curve is also $\fqs$-maximal.
For some decades, ideas exploring this result were used to obtain examples of maximal curves, most of them were covered by the famous $\mathbb{F}_{q^2}$-maximal Hermitian curve. In 2009, Giulietti and Korchmáros provided in \cite{GK2009} the first example of a $\mathbb{F}_{q^6}$-maximal curve, nowadays referred to as the $GK$-curve, which is not covered by the Hermitian curve over $\fqc$, for any $q > 2$. In the same year, Garcia, G\"{u}neri and Stichtenoth presented a generalization of the $GK$-curve \cite{GGS2010}, now known as the $GGS$-curve, that is, a maximal curve over $\mathbb{F}_{q^{2r}}$ for $r\geq 3$ odd and isomorphic to the $GK$-curve for $r=3$. 

For $r \geq 3$ odd, consider the $\mathbb{F}_{q^{2r}}$ maximal curve covered by the $GGS$-curve (see \cite{ABQ2009} and \cite{GGS2010}) defined by the affine equation
\begin{equation} \label{eq: X}
		\mathcal{\cX}: \quad y^{q^2}-y=x^{\tfrac{q^r+1}{q+1}}.
	\end{equation}
with genus 
	$$g(\cX) = \tfrac 12 {(q-1)(q^n-q)}.$$ 	
	This curve defines an elementary abelian $p$-extension, so it ramifies only  at the  place $Q_\infty$ over  $P_\infty$ in $\mathbb{F}_{q^{2r}}(x)$ with even different exponent 
	$$d(Q_\infty|P_\infty) = (q^2-1) \tfrac{q^r+q+2}{q+1} = (q-1)(q^r+q+2)$$ 
by \eqref{expdif}. So we can apply Theorem \ref{teolevantado} and lift an iso-dual code over $ \mathbb{F}_{q^{2r}}(x)$.

\begin{proposition} \label{prop: GGS}
For $q$  and $r \geq 3$ both odd there exist a lifted iso-dual AG-code on the function field of the  curve  $\cX$ over $\mathbb{F}_{q^{2r}}$ 
with parameters satisfying
	\begin{equation} \label{eq: parameters GGS}
		\big[ q^2(q^r+1)(q^{r-1}-1), \;  \tfrac 12 q^2(q^r+1)(q^{r-1}-1), \; \ge \tfrac 12 ( q^{2 r + 1} -q^{r + 2} + q^r - q - 2) \big].
	\end{equation}
\end{proposition}

\begin{proof}
Let $n=(q^r+1)(q^{r-1}-1)$ and let $P_1, \dots, P_n$ be the $\mathbb{F}_{q^{2r}}$-rational places splitting in the extension $\mathbb{F}_{q^{2r}}(\cX) / \mathbb{F}_{q^{2r}}(x)$, see   \cite[Lemma 2]{ABQ2009}. The code $C_\cL(D,G)$ with 
	$$D=P_1+\cdots+P_n \qquad \text{ and } \qquad G=\tfrac 12 (n-2) P_\infty$$ 
is a rational iso-dual code by Proposition \ref{prop:isodual} since $n$ is even. Thus, we are in the conditions of Theorem~\ref{teolevantado} and then $\mathfrak{L}_{\cX / \mathbb{F}_{q^{2r}}(x)}(C_\cL(D,G))$ is an iso-dual code over $\cX $ with parameters as stated in \eqref{eq: parameters Sz}, by Corollary \ref{parameters}.
\end{proof}

\section{Lifting iso-dual codes over the Hermitian function field} \label{sec: from hermitian}
In order to lift an iso-dual code in an extension of function fields, some technical requirements are given in Theorem \ref{teolevantado}. We start this section proposing an algebraic curve over $\fqs$ satisfying the requirements. 
After that, we are going to lift a certain iso-dual code.

Consider the algebraic variety $\cX$ over $\fqs$ defined by the equations 
\begin{equation} \label{curveX}
\cX: \quad 
\begin{cases}
z^{q+1}=y^q+y,\\
y^{q+1}=x^q+x.
\end{cases}
\end{equation}
We now prove that $\cX$ is an absolutely irreducible curve over $\fqs$ and compute the genus and the number of $\ff_{q^2}$-rational points of $\cX$.

\begin{proposition} \label{propo_nuevaX}
	$\cX$ is an absolutely irreducible curve over $\fqs$, has genus $g(\cX)=q^3-q$ and its number of rational points over $\ff_{q^2}$ is $\#\cX(\ff_{q^2})=q^4+1$.
\end{proposition}

\begin{proof}
For simplicity, we put $\cM=\fqs(x,y,z)$ and $\cF=\fqs(x,y)$. 
The first equation 
$y^{q+1}=x^q+x$ 
defines the maximal Hermitian function field $\cF$ of genus $\tfrac 12q(q-1)$. We denote by $P_\alpha$ the place in $\fqs(x)$ associated to the zero of $x-\a$ and by $P_\infty$ the pole of $x$. In the extension $\cF/\ff_{q^2}(x)$ we have the following well-known structure for the rational places: each $P_\a$ for $\a^q+\a=0$ is totally ramified and we denote by $Q_{\a, 0}$ be the only place over it. For $\a \in \fqs $ that is not a root of $x^q+x =0$ the place $P_\a $ splits and we denote by $Q_{\a, \b_i}$ the $q+1$ places over it, where $\b_i^{q+1}=\a^q+\a$.

We notice that $\cX$ is an absolutely irreducible curve since $Q_\infty$ is totally ramified in $\cM/\fqs(x)$. Now we prove that $\cM/\cF$ defines a Kummer extension. We start by computing the divisor of $y^q+y$ in $\cF$. Clearly $Q_\infty$ and $Q_{\alpha,0}$ are totally ramified in $\cM/\cF$. Consider the sets
\begin{align*}
 & S_0=\{\alpha\in \mathbb{F}_{q^2}: \alpha^q+\alpha= 0\},\\[1mm]
 & S_1=\{\alpha\in \mathbb{F}_{q^2}: \alpha^q+\alpha\neq 0 \text { and } x^2+\alpha^q+\alpha \, \in \mathbb{F}_{q}[x] \text { is irreducible}\}, %\text{ and}  
 \\[1mm] 
 & S_2=\{\alpha\in \mathbb{F}_{q^2}: \alpha^q+\alpha\neq 0 \text { and } x^2+\alpha^q+\alpha \, \in \mathbb{F}_{q}[x] \text { is reducible}\}.
\end{align*}
 Let $Q_{\a, \b}$ be a zero of $y^q+y$ in $\cF$, with $\alpha\neq 0$. If $\beta =0$ we have that $\alpha \in  S_0$. From now on we consider $\beta \neq 0$. Then, it is a simple zero and we have that $\b^q=-\b$ and $\b^{q+1}=\a^q+\a $ and hence $\b^2-\a^q-\a=0$ and 
$\b \in \fqs \smallsetminus \fq $. We conclude that $\a \in S_1$, i.e.  $\a \in \fqs$ such that $x^2+\a^q+\a$ is an irreducible polynomial over $\fq$ with two distinct roots $\b_\a$ and  $\b_\a^q$ in $\fqs \smallsetminus \fq$. We denote by $Q_{\a, \b_\a}$ and $Q_{\a, \b_{\a}^q}$ these places in $\cF$.  Moreover, since $Q_\infty$ is the only pole of $y$ of order $q$ we have that 
\begin{equation}\label{diff}
( y^q+y)^\cF =\sum_{\a \in S_0} Q_{\a,0} + 
			  \sum_{\a \in S_1 } (Q_{\a, \b_\a} + Q_{\a, \b_{\a}^q}) - q^2Q_\infty .
\end{equation}

Hence, we conclude that the extension $\cM/\cF$ defines a Kummer extension. By Kummer theory, the places $Q_{\a, \b_\a}$ and $Q_{\a, \b_{\a}^q}$ are totally ramified in the extension $\cM/\cF$. Since $|S_1|=\tfrac 12(q^2-q)$, we have a total of $q^2+1$ totally ramified places in the extension $\cM/\cF$.  In Figure~\ref{figu0} we summarize the ramification in the curve $\cX$.

 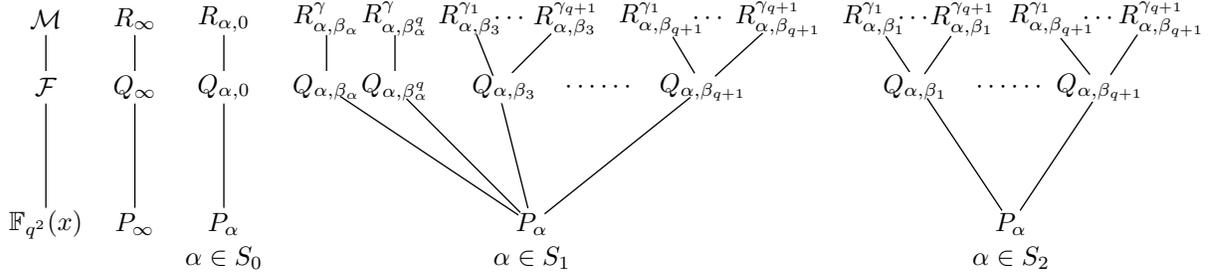
\begin{figure}[h!t]
        \begin{center}
  \begin{tikzpicture}[scale=.91]
 \draw[line width=0.5 pt](-3.1,0)--(-3.1,3);
   \draw[white,  fill=white](-3.1,0) circle (0.23 cm);
   \draw[white,  fill=white](-3.1,2) circle (0.23 cm);
\draw[white,  fill=white](-3.1, 3) circle (0.23 cm);
%            \draw[white,  fill=white](-0.9,3.7) circle (0.23 cm);
%               \draw[white,  fill=white](-0.9,4.7) circle (0.23 cm);
\node at(-3.1,0){\footnotesize{$\mathbb{F}_{q^2}(x)$}};
\node at(-3.1,2){\footnotesize{$\mathcal{F}$}};
\node at(-3.1,3){\footnotesize{$\mathcal{M}$}};
 \draw[line width=0.5 pt](-1.8,0)--(-1.8,3);
   \draw[white,  fill=white](-1.8,0) circle (0.23 cm);
   \draw[white,  fill=white](-1.8,2) circle (0.23 cm);
\draw[white,  fill=white](-1.8, 3) circle (0.23 cm);
\node at(-1.8,0){\footnotesize{$P_\infty$}};
\node at(-1.8,2){\footnotesize{$Q_\infty$}};
\node at(-1.8,3){\footnotesize{$R_\infty$}};
 \draw[line width=0.5 pt](-0.5,0)--(-0.5,3);
   \draw[white,  fill=white](-0.5,0) circle (0.23 cm);
   \draw[white,  fill=white](-0.5,2) circle (0.23 cm);
\draw[white,  fill=white](-0.5, 3) circle (0.23 cm);
\node at(-0.5,0){\footnotesize{$P_\alpha$}};
   \node at(-0.5,-0.5){\footnotesize{$\alpha \in S_0$}};
%\node at(-0.9,1.2){\footnotesize{$F_1$}};
%\node at(-0.9,2.2){\footnotesize{$F_2$}};
\node at(-0.5,2){\footnotesize{$Q_{\alpha,0}$}};
\node at(-0.5,3){\footnotesize{$R_{\alpha,0}$}};
     \draw[line width=0.5 pt](4,0)--(1,2);
        \draw[line width=0.5 pt](1,2)--(1,3);
     \draw[line width=0.5 pt](2,2)--(2,3);
                     \draw[line width=0.5 pt](4,0)--(3.5,2);
                     \draw[line width=0.5 pt](3.5,2)--(4.5,3);
                     \draw[line width=0.5 pt](3.5,2)--(3.1,3);
                       \draw[line width=0.5 pt](4,0)--(2,2);
                       \draw[line width=0.5 pt](4,0)--(6.5,2);
                     \draw[line width=0.5 pt](6.5,2)--(7.6,3);
                       \draw[line width=0.5 pt](6.5,2)--(5.9,3);

        \draw[white,  fill=white](4.5,3) circle (0.3 cm);
  \draw[white,  fill=white](1,2) circle (0.23 cm);
    \draw[white,  fill=white](1,3) circle (0.23 cm);  
  \draw[white,  fill=white](2,3) circle (0.23 cm);    
    \draw[white,  fill=white](2,2) circle (0.23 cm);     
   \draw[white,  fill=white](3.1,3) circle (0.23 cm); 
     \draw[white,  fill=white](3.6,2) circle (0.23 cm);
     \draw[white,  fill=white](7.6,3) circle (0.3 cm);   
   \draw[white,  fill=white](5.9,3) circle (0.3 cm); 
     \draw[white,  fill=white](6.5,2) circle (0.23 cm);
 \draw[white,  fill=white](4,0) circle (0.23 cm);
   \node at(4,0){\footnotesize{$P_\alpha$}};
 \node at(1,2){\footnotesize{$Q_{\alpha, \beta_\alpha}$}};
   \node at(5,2){\footnotesize{$\cdots\cdots$}};  
 \node at(3.6,2){\footnotesize{$Q_{\alpha,\beta_3}$}};         
  \node at(1,3){\footnotesize{$R_{\alpha, \beta_\alpha}^{\gamma}$}};      
 \node at(2,3){\footnotesize{$R_{\alpha, \beta_\alpha^q}^{\gamma}$}};  
  \node at(2,2){\footnotesize{$Q_{\alpha, \beta_\alpha^q}$}};  
  \node at(3.7,3){\footnotesize{$\cdots$}};  
    \node at(6.6,3){\footnotesize{$\cdots$}};  
 \node at(3.1,3){\footnotesize{$R_{\alpha, \beta_3}^{\gamma_1}$}};      
    \node at(4.5,3){\footnotesize{$R_{\alpha, \beta_3}^{\gamma_{q+1}}$}};  
     \node at(6.5,2){\footnotesize{$Q_{\alpha,\beta_{q+1}}$}};    
     \node at(5.9,3){\footnotesize{$R_{\alpha,\beta_{q+1}}^{\gamma_1}$}};      
    \node at(7.6,3){\footnotesize{$R_{\alpha,\beta_{q+1}}^{\gamma_{q+1}}$}};  
            \node at(4,-0.5){\footnotesize{$\alpha \in S_1$}};

           \draw[line width=0.5 pt](11,0)--(9,3);
     \draw[line width=0.5 pt](9.6,2)--(10.3,3);
                     \draw[line width=0.5 pt](11,0)--(13,3);
                     \draw[line width=0.5 pt](12.3,2)--(11.5,3);
        \draw[white,  fill=white](12,0) circle (0.23 cm);
  \draw[white,  fill=white](9.6,2) circle (0.23 cm);
    \draw[white,  fill=white](9,3) circle (0.3 cm);  
  \draw[white,  fill=white](10.3,3) circle (0.3 cm);    
 \draw[white,  fill=white](12.9,3) circle (0.23 cm);   
   \draw[white,  fill=white](11.6,3) circle (0.3 cm); 
     \draw[white,  fill=white](12.3,2) circle (0.23 cm);
 \draw[white,  fill=white](11,0) circle (0.23 cm);
    \node at(11,0){\footnotesize{$P_\alpha$}};
        \node at(11,-0.5){\footnotesize{$\alpha \in S_2$}};
      %    \node at(11,-1){\footnotesize{$x^2+\alpha^q+\alpha$ red. over $\mathbb{F}_q$}};
 \node at(9.6,2){\footnotesize{$Q_{\alpha,\beta_1}$}};
   \node at(11,2){\footnotesize{$\cdots\cdots$}};  
 \node at(12.3,2){\footnotesize{$Q_{\alpha,\beta_{q+1}}$}};         
  \node at(9,3){\footnotesize{$R_{\alpha,\beta_1}^{\gamma_1}$}};      
 \node at(10.3,3){\footnotesize{$R_{\alpha,\beta_1}^{\gamma_{q+1}}$}};  
  \node at(9.6,3){\footnotesize{$\cdots$}};  
    \node at(11.6,3){\footnotesize{$R_{\alpha,\beta_{q+1}}^{\gamma_1}$}};      
 \node at(13.2,3){\footnotesize{$R_{\alpha,\beta_{q+1}}^{\gamma_{q+1}}$}};  
  \node at(12.3,3){\footnotesize{$\cdots$}};  
 \end{tikzpicture}
  \caption{Decomposition of places of $\mathbb{F}_{q^2}(x)$ in $\mathcal{M}$.} \label{figu0}
\end{center}\end{figure}

From \eqref{diff} we obtain that the degree of the different $\text{Diff}(\cM/\cF)$ is $q^3+q$, since there are $q^2+1$ totally ramified places in the extension $\cM/\cF$. 
Now, since $[\cM:\cF]=q+1$, the Riemann-Hurwitz genus formula yields
\begin{align*}
2g(\cM)-2 = 2 g(\cF) -2 + \deg(\text{Diff}(\cM/\cF)) = (q+1)(q^2-q-2)+ q^3+q.
\end{align*}
That is $2g(\cM)-2 = 2(q^3-q-1)$, from which we get $g(\cX)=g(\cM)=q^3-q$.

Now we compute the number of $\fqs$-rational points on the curve $\cX$. The place at infinity $R_\infty$ is rational. 
Clearly, there are $q$ places of the form $R_{\alpha,0}$ with $\alpha \in S_0$. 
Let $R_{\a, \b}^ {\g}$ be a rational place in $\cM$ over $Q_{\a, \b}$ in $\cF$ for $\beta \neq 0$, then 
	$$\g^{q+1}=\b^q+\b \qquad \text{and} \qquad \b^{q+1}=\a^q+\a \neq 0.$$ 
We consider two cases.

\noindent \textit{Case 1}: 
If $\b^q+\b=0$, then $\b^{q+1}=\a^q+\a$ implies $-\b^2=\a^q+\a, \b \in \fqs \smallsetminus \fq $. So $x^2+\a^q+\a$ is irreducible over $\fq$ and we obtain (as before) that the two places $Q_{\a, \b_\a}$ and $Q_{\a, \b_{\a}^q}$ are both totally ramified in $\cM/\cF$, all the other $Q_{\a, \b_i}$ split in $\cM/\cF$. 

\noindent \textit{Case 2}: If $\b^q+\b \neq 0$ then the equation $z^{q+1}-\b^q-\b$ factors into $q+1$  factors of degree one over $\fqs$. Hence the place   $Q_{\a, \b}$ splits in the extension $\cM/\cF$. 

Considering the two cases we have a total of 
	$$ 2 \tfrac{(q^2-q)}{2}+(q+1) \big( (q-1) \tfrac{q^2-q}{2} + (q+1)(\tfrac{q^2-q}2)  \big) = (q^2-q)+(q+1)(q^2-q)q$$ 
rational places over $\fqs$ (recall that $\frac{q^2-q}{2}$ is the number of $\a$ in $\fqs$ such that $\a^q+\a \neq0$ and $\a^q+\a$ is a square in $\fq$, or not a square in $\fq$).

Summarizing, the curve $\cX$ has in total  
$$1 + q + (q^2-q) + (q^2-q)(q+1)q = q^4+1$$
rational points over $\fqs$.
\end{proof}

We now present an iso-dual code defined in $\mathcal{M}=\mathbb{F}_{q^2}(\mathcal{X})$ for certain values of $q$. We will also consider the intermediate function field $\mathcal{F}$ as in the proof of Proposition \ref{propo_nuevaX}. 

\begin{theorem} \label{thm: isodual Herm}
For any $q=2^s$ with $s>1$ 
% \lu{a condição aqui não pode ser OR pois precisamos que todas as ramificações sejam ímpares e ela é  $(q+1)^2$, tem que ser somente $q$ par mesmo} 
there is a lifted iso-dual code over $\mathcal{M}=\mathbb{F}_{q^2}(\mathcal{X})$, where $\cX$ is as in \eqref{curveX}, with parameters 
$$ n = \tfrac 12(q^2-q)(q+1)^2, \quad k = \tfrac 14(q^2-q)(q+1)^2, \quad \text{and} \quad d \ge \tfrac 14(q^4-q^3-q^2+3q+4).$$ 
\end{theorem}

\begin{proof}
Taking the divisors 
	$$D=\sum_{\alpha \in S_2} P_\alpha \qquad \text{and} \qquad G= \tfrac 14(q^2-q-4) P_\infty$$ 
in the rational function field $\mathbb{F}_{q^2}(x)$ we have, by \eqref{isorat} in Proposition \ref{prop:isodual}, that the code $C_\cL(D, G)$ is a rational iso-dual code with parameters 
	$$n = \tfrac 12(q^2-q) \text{ even}, \qquad k = \tfrac 14(q^2-q) \qquad \text{and} \qquad  d \ge \tfrac 12(q^2-q+1).$$
This code can be lifted to $\mathcal{F}$ using Theorem~\ref{teolevantado}  to an iso-dual code $\mathfrak{L}_{\cF/\mathbb{F}_{q^2}(x)}(C_\cL(D,G))=C_\cL(D^{\cF},G^{\cF})$ of length and dimension given by 
	$$ n^{\mathcal{F}} = \tfrac 12(q^2-q)(q+1)  \qquad \text{and} \qquad k^\cF = \tfrac 14(q^2-q)(q+1).$$ 

Since this code is iso-dual we have 
	%$$G^{\mathcal{F}} = \con_{\cM/\cF}(G) + \tfrac 12 \diff(\mathcal{F}/\mathbb{F}_{q^2}) = \tfrac 14 (q+1) (q^2-q-4) \, Q_\infty + \tfrac{q}{2} \Big( Q_\infty + \sum_{\alpha\in S_0} Q_{\alpha,0} \Big)$$ 
 $\deg(G^{\mathcal{F}})=\frac{1}{4}(q^3+2q^2-3q-4)$ and 
	$$d^\cF \ge n^{\mathcal{F}}-\deg(G^{\mathcal{F}}) = \tfrac 14 (q^3-2q^2+q+4).$$ 
Now,  from Remark \ref{canonicallifting}  we are again in the conditions of Theorem~\ref{teolevantado}, and then the lifted code   
$$\mathfrak{L}_{\cM/\cF}(\mathfrak{L}_{\cF/\mathbb{F}_{q^2}(x)}(C_\cL(D,G)))=\mathfrak{L}_{\cM/\cF}(C_\cL(D^{\cF},G^{\cF}))=C_\cL(D^{\cM},G^{\cM})$$
is an iso-dual code over $\mathcal{M}$, where 
	$$n^{\mathcal{M}} = \tfrac 12(q^2-q)(q+1)^2  \qquad \text{and} \qquad  
	  k^\mathcal{M} = \tfrac 14(q^2-q)(q+1)^2.$$ 
We also have $\deg(G^{\mathcal{M}})=
	\tfrac 14 (q^4+5q^3-q^2-5q-4)$ and, therefore,
	$$d^\mathcal{M} \ge n^\cM-\deg(G^\cM) = \tfrac 14(q^4-q^3-q^2+3q+4),$$ 
as we wanted to show. \end{proof}

\section{Binary and ternary cyclotomic iso-dual codes} \label{sec:cyclotomic}
Some subfields of a cyclotomic function field has been used by Quebbemann in \cite{Que88} to give examples of long AG-codes (called cyclotomic Goppa codes). Here, we will  construct long binary and ternary iso-dual  AG-codes by using an alternative approach introduced in \cite{NX1996} to produce optimal function fields over $\ff_2$. In view of the terminology used by Quebbemann in \cite{Que88}, we will call these codes binary and ternary cyclotomic iso-dual codes.

We follow the presentation of Hayes \cite{Ha74} of cyclotomic function fields (see also Section~3.2 of \cite{NX2001} where a summary of the main results of Hayes using divisors in additive notation is presented). Let $R=\fq[x]$ be the polynomial ring over $\fq$ and let $\bar{\cF}$ be an algebraic closure of the rational function field $\cF=\fq(x)$. Let $\varphi$ the $\fq$-vector space endomorphism of $\bar{\cF}$ defined as
	$$\varphi(u)=u^q+xu,$$ 
for all $u\in  \bar{\cF}$. Then $\varphi$ induces a ring homomorphism from $R$ to $\mathrm{End}_{\fq}(\bar{\cF})$ by $f\mapsto f(\varphi)$; that is, if $f(x)=\sum_i a_i x^i$ then $f(\varphi)=\sum_i a_i \varphi^i$. 
This, in turn, allows  to define an $R$-module structure on $\bar{\cF}$ by defining an action of $R$ on $\bar{\cF}$ as
	$$u^{f}=f(\varphi)(u),$$
for $f\in R$ and $u\in  \bar{\cF}$.

Let $f\in R$. By considering the submodule of $f$-torsion points of this action
\[\Lambda_f=\{u\in  \bar{\cF}:u^f=0\},\]
we have the field $\cF(\Lambda_f)$ generated over $\cF$ by the elements of $\Lambda_f$. This field is a finite abelian extension of $\cF$ and its Galois group 
is isomorphic to 
the group of units of the quotient ring $R/(f)$, where $(f)$ denotes the principal ideal of $R$ generated by $f$, that is 
$$ \mathrm{Gal}(\cF(\Lambda_f)/\cF) \simeq \left(R/(f)\right)^*.$$
%\rp{By this reason, $\cF(\Lambda_f)$ is called a cyclotomic function field.} 

Assume that $h\in R$ is monic and irreducible over $\fq$. As before, we denote by $P_h$ the place of $\cF$ associated to $h$. The unique automorphism $\sigma_{ \bar{h}}\in \mathrm{Gal}(\cF(\Lambda_f)/\cF)$ determined by its residual class $ \bar{h}\in \left(R/(f)\right)^*$ acts as $\sigma_{ \bar{h}}(u)=u^h$ for all $u\in \Lambda_f$. 
In particular, if $h$ does not divide $f$, the Artin symbol 
$$ \left[\tfrac{\cF(\Lambda_f)/\cF}{P_h}\right] $$ 
of the place $P_h$ is the automorphism $\sigma_{ \bar{h}}$. Therefore, if $H$ is the subgroup of $\left(R/(f)\right)^*$ generated by $ \bar{h}$, where $h\in R$ is irreducible and does not divide $f$, then $P_h$ splits completely in the fixed subfield 
\begin{equation} \label{eq: Kf}
	\cK=\cF(\Lambda_{f})^H
\end{equation}
of $\cF(\Lambda_f)$ (see, for instance, Proposition 1.4.12 of \cite{NX2001}). Thus, with a suitable choice of such a polynomial $h$ we can find many rational places in a subfield of $\cF(\Lambda_f)$ whose genus and degree can be explicitly computed. Some optimal function fields over $\ff_2$ were found in \cite{NX1996} in this way.

Suppose now that $f\in R$ is monic of degree $d$ and irreducible over $\ff_q$. We have that  			
$$\cF(\Lambda_{f^n})/\cF$$ 
is a finite abelian field extension of degree 
	$$q^{dn}-q^{d(n-1)}=q^{d(n-1)}(q^d-1),$$ 
where the places $P_f$ and $P_\infty$ (the only pole of $x$ in $\cF$) are the only places of $\cF$ that can be ramified in $\cF(\Lambda_{f^n})$. In fact, $P_f$ is always totally ramified in $\cF(\Lambda_{f^n})$ and the  place $P_\infty$ is totally ramified in $\cF(\Lambda_f)$ and then it splits completely in $\cF(\Lambda_{f^n})/\cF(\Lambda_f)$. In terms of ramification indices, we have that if we denote by $Q_\infty$  the only place of $\cF(\Lambda_f)$ lying over $P_\infty$ then $Q_\infty$ is a rational place, $e(Q_\infty|P_\infty)=q-1$ and there are exactly 			
$$r=q^{d(n-1)}(q^{d-1}+q^{d-2}+\cdots+q+1)$$ 
places of $\cF(\Lambda_{f^n})$ lying over $Q_\infty$ and they are all rational places.  For the place $P_f$, the only place of $\cF(\Lambda_{f^n})$ is denoted by $R_f$ and its restriction to $\cF(\Lambda_f)$ is denoted by $Q_f$. This situation is illustrated in Figure \ref{figcyclotomic1} below.
\begin{figure}[h!t] %\label{figCyc1}
\begin{center}
	\begin{tikzpicture}
		\coordinate (a) at (0,0);
		\coordinate (b) at (0,1.3);
		\coordinate (c) at (0,2.5);
		\draw (a)--(c);
		\filldraw[white] (a) circle [radius=0.1];
		\filldraw[white] (b) circle [radius=0.3];
		\filldraw[white] (c) circle [radius=0.1];
		\node at (a) [below] {\scalebox{.75}{$\cF$}};
		\node at (b)  {\scalebox{.75}{$\cF(\Lambda_f)$}};
		\node at (c) [above] {\scalebox{.75}{$\cF(\Lambda_{f^n})$}};
		\node at (-.6,.5){\scalebox{.75}{$q-1$}};
		\node at (-.5,2){\scalebox{.75}{$r$}};
		\coordinate (d) at (1.5,0);
		\coordinate (e) at (1.5,1.3);
		\coordinate (f) at (1.5,2.5);
		\draw (d)--(f);
		\filldraw[white] (d) circle [radius=0.1];
		\filldraw[white] (e) circle [radius=0.3];
		\filldraw[white] (f) circle [radius=0.1];
		\node at (d) [below] {\scalebox{.75}{$P_f$}};
		\node at (e)  {\scalebox{.75}{$Q_f$}};
		\node at (f) [above] {\scalebox{.75}{$R_f$}};
		\node at (2.3,.5){\scalebox{.75}{$e=q-1$}};
		\node at (2,2){\scalebox{.75}{$e=r$}};
		\coordinate (g) at (3.5,0);
		\coordinate (h) at (3.5,1.3);
		\coordinate (i) at (3.5,2.5);
		\draw (g)--(h);
		\draw (h)--(3,2.5);
		\draw (h)--(4,2.5);
		\filldraw[white] (g) circle [radius=0.1];
		\filldraw[white] (h) circle [radius=0.3];
		\filldraw[white] (i) circle [radius=0.1];
		\node at (g) [below] {\scalebox{.75}{$P_\infty$}};
		\node at (h)  {\scalebox{.75}{$Q_\infty$}};
		\node at (i) [above] {\scalebox{.75}{$R_\infty^1 \cdots\cdots R_\infty^r$}};
		\node at (4.5,.5){\scalebox{.75}{$e=q-1$}};
		\node at (4.5,2){\scalebox{.75}{$e=1$}};	
	\end{tikzpicture}
	\caption{Decompositions of $P_f$ and $P_\infty$.} \label{figcyclotomic1}
\end{center}
\end{figure}
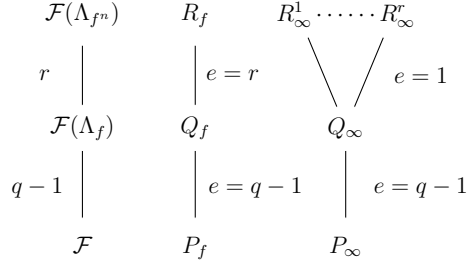

From now on, we consider the particular case $f=x$ in \eqref{eq: Kf} and then $\cF(\Lambda_{x^n})/\cF$ 
is a cyclotomic extension of $\cF$ of degree $q^{n-1}(q-1)$. We know that 
\begin{equation}\label{genusf=x}
	2g(\cF(\Lambda_{x^n}))-2=q^{n-1}(n(q-1)-q-1)
\end{equation}
(see Corollary 4.2 of \cite{Ha74}) and thus the function field $\cF(\Lambda_x)$ is, in fact, a rational function field over $\fq$ (see Proposition 1.6.3 of \cite{Stichbook09}). 

We will define a cyclic subgroup $H$ of $\left(R/(x^n)\right)^*$ of order $q^m$ generated by a suitable residual class  
$\overline{x-\alpha}$ for some $\alpha\in \ff_q$. With this choice of $H$ we have that the subfield 										
	$$\mathcal{K}_n = \cF(\Lambda_{x^n})^H$$ 
defines a cyclic extension $\cF(\Lambda_{x^n})/\mathcal{K}_n$ of degree $q^m$ and we will try to lift a rational iso-dual AG-code to an iso-dual AG-code over $\mathcal{K}_n$.

For an estimate of the minimum distance of these liftings we need to compute the genus of $\mathcal{K}_n$ according to Corollary \ref{parameters}. The case $q=2$ of item $(b)$ of the following proposition was proved in Theorem 2 of \cite{NX1996}. 

\begin{proposition} \label{Kn}
Let $n,q\in \N$ with $q$ a prime power and put $m=\lceil \log_q(n)\rceil$. Let $H$ be the subgroup of $\left(R/(x^n)\right)^*$  generated by the residual class $ \overline{x+1}$. Let $\mathcal{K}_n=\cF(\Lambda_{x^n})^H$ be the fixed subfield of $\cF(\Lambda_{x^n})$ by $H$ and let $S_x=R_x\cap\cK_n$ be the restriction to $\cK_n$ of the only place $R_x$ of $\cF(\Lambda_{x^n})$ lying over $P_x$. Then, we have the following. 
\begin{enumerate}[$(a)$]
\item $H$ is a cyclic group of order $q^m$. In particular, the extension $\cF(\Lambda_{x^n})/\cK_n$ is cyclic of degree $q^m$, the extension $\cK_n/\cF$ is of degree $q^{n-m-1}(q-1)$ and the place $S_x$ is a rational place of $\cK_n$ which is totally ramified in $\cF(\Lambda_{x^n})$ (see Figure \ref{figcyclotomic2} below). \sk 

\item If $\cF(\Lambda_x)\subset \cK_n$, then the extension $\mathcal{K}_n/\cF(\Lambda_x)$ is unramified outside the place $Q_x$ and the genus of $\cK_n$ is %given by 
\begin{equation} \label{Kgenus}
	g(\mathcal{K}_n)=  \tfrac 12 q^{-m} \Big(q^{n-1}(n(q-1)-q-1)-\sum_{i=1}^{q^m-1}q^{e_i} \Big) +1,
\end{equation}
where $e_i$ denotes the least integer such that $p=\emph{Char}(\fq)$ does not divide the binomial number $\binom{i}{j}$ for $1\leq j\leq i<n$.
%\item Suppose that the extension $\cF(\Lambda_{x^n})/\mathcal{K}_n$ is unramified  outside the place $S_x$. Let $e_i$ be the least integer such that $p=\textrm{Char}(\fq)$ does not divide the binomial number $\binom{i}{j}$ for $1\leq j\leq i<n$. Then
%  	\begin{equation}\label{Kgenus}
%  		2g(\mathcal{K}_n)-2=q^{-m} \Big(q^{n-1}(n(q-1)-q-1)-\sum_{i=1}^{q^m-1}q^{e_i} \Big).
%  	\end{equation}
\end{enumerate}
\end{proposition}

\goodbreak 

\begin{proof}
$(a)$ We clearly have that $m$ is the unique integer such that $q^{m-1}<n\leq q^m$. Since 
	\[ (x+1)^{q^m}=x^{q^m}+1\equiv 1 \pmod{x^n}, \]
we see not only that $\overline{x+1}\in \left(R/(x^n)\right)^*$ but also that $H$ is a cyclic group of order $q^m$ because of the choice of the integer $m$. Thus, $\cF(\Lambda_{x^n})/\cK_n$ is a cyclic extension of degree $q^m$. 

Since $R_x$ is a rational place so is $S_x$ and since $R_x$ is the only place of $\cF(\Lambda_{x^n})$ lying over $S_x$ we must have that $e(R_x|S_x)=q^m$, that is $S_x$ is totally ramified in $\cF(\Lambda_{x^n})$. We sketch this situation in the following picture:
\begin{figure}[h!t] \label{figCyc2}
\begin{center}
		\begin{tikzpicture}
			\coordinate (a) at (0,0);
			\coordinate (b) at (0,1.3);
			\coordinate (c) at (0,2.5);
			\draw (a)--(c);
			\filldraw[white] (a) circle [radius=0.1];
			\filldraw[white] (b) circle [radius=0.3];
			\filldraw[white] (c) circle [radius=0.1];
			\node at (a) [below] {\scalebox{.75}{$\cF$}};
			\node at (b)  {\scalebox{.75}{$\mathcal{K}_n$}};
			\node at (c) [above] {\scalebox{.75}{$\cF(\Lambda_{x^n})$}};
			\node at (-1.2,.5){\scalebox{.75}{$q^{n-m-1}(q-1)$}};
			\draw (-.6,2.8) arc[start angle=158, end angle=200,radius=2.1cm] ;
			\node at (-1.5,2.3){\scalebox{.75}{Cyclic of}};
			\node at (-1.5,1.9){\scalebox{.75}{degree $q^m$}};
			\coordinate (d) at (1.5,0);
			\coordinate (e) at (1.5,1.3);
			\coordinate (f) at (1.5,2.5);
			\draw (d)--(f);
			\filldraw[white] (d) circle [radius=0.1];
			\filldraw[white] (e) circle [radius=0.3];
			\filldraw[white] (f) circle [radius=0.1];
			\node at (d) [below] {\scalebox{.75}{$P_x$}};
			\node at (e)  {\scalebox{.75}{$S_x$}};
			\node at (f) [above] {\scalebox{.75}{$R_x$}};
			\node at (3.1,.5){\scalebox{.75}{$e=q^{n-m-1}(q-1)$}};
			\node at (2.3,2){\scalebox{.75}{$e=q^m$}};
		\end{tikzpicture}
	\caption{Ramification of $S_x$.} \label{figcyclotomic2}
	\end{center}
\end{figure}	

\noindent	
$(b)$ Assume that $\cF(\Lambda_x)\subset \cK_n$. From the ramification situation described in Figure \ref{figcyclotomic1}, we immediately see that the extension $\mathcal{K}_n/\cF(\Lambda_x)$ is unramified outside the place $Q_x$. In particular, since $S_x$ is the only place of $\cK_n$ lying over $Q_x$, the extension $\cF(\Lambda_{x^n})/\cK_n$ is unramified outside the place $S_x$ and from \eqref{oneramified} we have that
	\[ 2g(\cF(\Lambda_{x^n}))-2 = (2g(\mathcal{K}_n)-2) q^m + d(R_x|S_x). \]
Now notice that from \eqref{genusf=x} we just need to compute the different exponent $d(R_x|S_x)$ to find the genus of $\mathcal{K}_n$. In order to do so, we recall that any root $\lambda$ of the polynomial $h(u)=u^{x^n}/u^{x^{n-1}}
$is a prime element of $R_x$, that is $\nu_{R_x}(\lambda)=1$ (Proposition 2.4 of \cite{Ha74}). Clearly the  minimal polynomial $g(u)$ of $\lambda$ over $\mathcal{K}_n$ is
	\[ g(u) = \prod_{\sigma\in H} (u-\sigma(\lambda)), \]
 so that from Proposition 3.5.12 of \cite{Stichbook09} we have that 
	\[ d(R_x|S_x) = \nu_{R_x}(g'(\lambda)). \] 
Since $H=\{\left( \overline{x+1}\right)^{i}:0\leq i<q^m\}$  we can write
	\[g'(\lambda) = \prod_{i=1}^{q^m-1}(\lambda-\lambda^{(x+1)^i}).\]
On the other hand (because $\lambda^{x^j}=0$ for $j\geq n$) we have that
	\[ \lambda^{(x+1)^i} = \sum_{j=0}^i \tbinom{i}{j}\lambda^{x^j} = \lambda+\sum_{j=1}^{i<n} \tbinom{i}{j}\lambda^{x^j},\]
so that
	\[g'(\lambda) = -\prod_{i=1}^{q^m-1} \sum_{j=1}^{i<n} \tbinom{i}{j}\lambda^{x^j}, \]
By induction on $j$, it is easy to see (since $\nu_{R_x}(\lambda)=1$) that $\nu_{R_x}(\lambda^{x^j})=q^j$ for $1\leq j\leq n-2$ and $\nu_{R_x}(\lambda^{x^{n-1}}) \ge q^{n-1}$. Then
	\[ \nu_{R_x} \Big(\sum_{j=1}^{i<n} \tbinom{i}{j} \lambda^{x^j} \Big) = q^{e_i}, \]
where $e_i$ is the least integer such that $p=\textrm{Char}(\fq)$ does not divide $\binom{i}{j}$ for $1\leq j\leq i$, and thus 
	\[ d(R_x|S_x) = \nu_{R_x}(g'(\lambda)) = \sum_{i=1}^{q^m-1}q^{e_i}. \]
Therefore 
\begin{equation} \label{eq: gKn 2}
	2g(\mathcal{K}_n)-2 = q^{-m} \Big(q^{n-1}(n(q-1)-q-1)-\sum_{i=1}^{q^m-1}q^{e_i} \Big),
\end{equation} 
from which \eqref{Kgenus} readily follows.
\end{proof}

We define now iso-dual AG-codes over $\mathcal{K}_n$ in the cases $q=2$ and $q=3$ by lifting to $\mathcal{K}_n$ some rational iso-dual codes.
\subsection{Binary cyclotomic iso-dual codes} Let $q=2$ and $\cF=\ff_2(x)$.
%Let us consider $q=2$ so that $\cF=\ff_2(x)$.
%Let $\cC=C(D,G)$ be an AG-code over $\cF=\ff_2(x)$. Then it has length $n\le 3$. If $\cC$ is isodual, by Prop.. $n=2$. 
%There are only 5 codes of length 2 over $\ff_2$: the trivial codes $0$ and $\ff_2^2$ and the codes of dimension 1, $C_1=\langle (1,0) \rangle$, $C_3=\langle (0,1) \rangle$ and $C_3=\langle (1,1) \rangle$. The codes $C_1$ and $C_2$ have distance $1$ and are $x$-isodual via $x=(1,0)$ and $(0,1)$ respectively. The code $C_3$ is the repetition code $Rep_2(2)$ with distance $d=2$, and it is self-dual. To apply the lifting procedure to these codes, we have to realize them as AG-codes.
%\begin{lemma} \label{lem: isodual l=2}
Since we have exactly three rational places in $\cF$ (namely $P_x$, $P_{x+1}$ and $P_\infty$), if we want to define an iso-dual code over $\cF$ of the form $C_\cL(D,G)$ which can be lifted to an iso-dual code over another field,  we are forced to consider $D$ as a divisor of $\cF$ whose support consists of exactly two rational places (the length must be even) and $G$ must be a  divisor of $\cF$ of degree zero. In the present situation the place $P_\infty$ splits completely in any cyclotomic extension of $\cF$ and so we define
\begin{equation} \label{eq: D cyc}
	D=P_{x+1}+P_\infty.
\end{equation}
Since we must have $\deg G=0$ (according to item ($c$) of Proposition \ref{prop:isodual}) we define   
\begin{equation} \label{eq: G cyc}
	G=P_{x^2+x+1}-2P_x,
\end{equation}
and thus $C_\cL(D,G)$ is an iso-dual AG-code over $\cF$. %, by Proposition \ref{prop:isodual}. 

This code is the repetition code $Rep_2(2)=\{(0,0), (1,1)\}$ with parameters $[2,1,2]$, which is actually self-dual.
	By lifting this code to a properly chosen subfield $\cK$ of a cyclotomic function field we will get a not trivial iso-dual AG-code over $\cK$.
	We now show that for every integer $n\geq 2$ there is a binary iso-dual AG-code defined over $\cK_n$.	
	\begin{theorem} \label{thm:binary}
		For any integer $n\geq 2$ the lifted code $\mathfrak{L}_{\mathcal{K}_n/\cF}(C_\cL(D,G))$, where $D$ and $G$ are as in \eqref{eq: D cyc} and \eqref{eq: G cyc}, is a binary iso-dual AG-code over $\mathcal{K}_n$ of length $2^{n-m}$ and dimension $2^{n-m-1}$, where $m=\lceil \log_2(n) \rceil$. 		
\end{theorem}

\begin{proof}
	%Define a cyclotomic extension $\cF(\Lambda_{f^n})$ of $\cF$ by taking $f=x$ and 
	First notice that when $q=2$ the rational place $P_\infty$ splits completely in $\cF(\Lambda_{x^n})$ and thus it splits completely in any subfield of $\cF(\Lambda_{x^n})$. In particular $P_\infty$ splits completely in $\cK_n
	$. The only ramified place in  $\cF(\Lambda_{x^n})/\cF$ is the place $P_x$ which is, in fact, totally ramified in $\cF(\Lambda_{x^n})/\cF$. Therefore  \[\textrm{Diff}(\cK_n/\cF)=d(S_x|P_x)S_x,\]
	and since $S_x$ is rational (see item ($a$) of Proposition \ref{Kn}) we have, from \eqref{oneramified},  that  the different exponent $d(S_x|P_x)$ is even.
	   
	On the other hand,  since $x+1$ is irreducible over $\ff_2$ and does not divide $x$, we have that $P_{x+1}$ splits completely in $\mathcal{K}_n$.
  Hence, we 	can apply Theorem \ref{teolevantado} in this situation with $D$ and $G$ as in \eqref{eq: D cyc} and \eqref{eq: G cyc} respectively, and we have that the lifted code $\mathfrak{L}_{\mathcal{K}_n/\cF}(C_\cL(D,G))$ is an iso-dual AG-code over $\mathcal{K}_n$.  In fact, since $[\cK_n:\cF]=2^{n-m-1}$ (see Figure \ref{figcyclotomic2} above),  $\mathfrak{L}_{\mathcal{K}_n/\cF}(C_\cL(D,G))$ is a binary iso-dual AG-code over $\mathcal{K}_n$ of length $2^{n-m}$ and dimension $2^{n-m-1}$. 
\end{proof}

\subsection{Ternary cyclotomic iso-dual codes} Assume now that $q=3$ so that $\cF=\ff_3(x)$. In this case we have four rational places $P_x$, $P_{x-1}$, $P_{x-2}$ and $P_\infty$ in $\cF$. We take again $f=x$ and, unlike the previous case in which $q=2$, we have now that $P_\infty$ is ramified in $\cF(\Lambda_{x^n})$. In fact, $P_\infty$ ramifies in $\cF(\Lambda_x)$ and then it splits completely into $3^{n-1}$ rational places of $\cF(\Lambda_{x^n})$ (see Figure \ref{figcyclotomic1} with $f=x$). Since in this case $e(R|P_\infty)=2$ for any place $R$ of $\cF(\Lambda_{x^n})$ lying over $P_\infty$ and $q=3$, each different exponent $d(R|P_\infty)=1$ and we can not use Theorem \ref{teolevantado} directly to lift an iso-dual rational AG-code over $\cF$ into a subextension of $\cF(\Lambda_{x^n})$. 

However, since $\cF(\Lambda_x)$ is a rational function field, if we show  that 
\[\cF(\Lambda_x)\subset\cK_n,\] 
then we will be able to use Theorem \ref{teolevantado} to lift  a rational iso-dual code over $\cF(\Lambda_x)$ to $\cK_n$ by showing that an even number of rational places of $\cF(\Lambda_x)$ split completely in $\cK_n$ (notice that if $\cF(\Lambda_x)\subset\cK_n$ then, from Theorem \ref{Kgenus}, we see that we are in the situation of formula \eqref{oneramified}, so that  the different exponent $d(S_x|Q_x)$ is even.)

We show now that $\cF(\Lambda_x)$ is also a subextension of $\mathcal{K}_n$. For any $u\in  \bar{\cF}$ the action is $u^x=u^3+xu$, and hence we see that
\[ \Lambda_x=\{u\in  \bar{\cF}: u^3+xu=0\}=\{0\}\cup\{u\in  \bar{\cF}:u^2+x=0\},\]
so that $\ff_3(x)=\ff_3(y^2)\subsetneq\ff_3(y)\subset \cF(\Lambda_x)$ where $y^2+x=0$. Since $\Lambda_x\subset\ff_3(y)$ we  conclude that $\cF(\Lambda_x)=\ff_3(y)$ and since
\[ y^{x+1}=(x+1)(\varphi+\mu_x)(y)=(\varphi+\mu_x)(y)+y=y^3+xy+y=y,\]
we also have that $\cF(\Lambda_x)$ is fixed by $H$ so that $\cF(\Lambda_x)\subset \mathcal{K}_n$. 

Now, since the rational place $P_{x-2}=P_{x+1}$ of $\cF$ splits completely into $2\cdot3^{n-m-1}$ rational places of $\mathcal{K}_n$, then $P_{x-2}$ splits completely into two rational places of $\cF(\Lambda_x)$ because $\cF(\Lambda_x)/\cF$ is a quadratic extension when $q=3$. 
In this way, we have three rational places of $\cF(\Lambda_x)$ splitting completely in $\mathcal{K}_n$, namely the two rational places of $\cF(\Lambda_x)$ lying over $P_{x-2}$, say $Q_{x-2}^1$ and $Q_{x-2}^2$, and $Q_\infty$. 

We define the divisors $D$ and $G$ of $\cF(\Lambda_x)$ as 
\begin{equation} \label{eq: D ter}
	D=Q_{x-2}^1+Q_{x-2}^2, %\qquad \text{and} \qquad
\end{equation} 
and $G=(y)$. %, respectively. 
From the equation $y^2+x=0$ we see that 
\begin{equation} \label{eq: G ter}
	G=Q_x-Q_\infty
\end{equation} 
so that $\supp(D)\cap\supp(G)=\varnothing$. Since $\deg G=0$ and $\cF(\Lambda_x)$ is a rational function field, we have from %\eqref{isorat} 
($c$) of Proposition \ref{prop:isodual} that $C_\cL(D,G)$ is an iso-dual AG-code over $\cF(\Lambda_x)$. Since  
\[ \mathrm{Diff}(\mathcal{K}_n/\cF(\Lambda_x))=d(S_x|Q_x) S_x,\]
and $d(S_x|Q_x)$ is even, we see that  the extension $\mathcal{K}_n/\cF(\Lambda_x)$ satisfies the conditions of Theorem \ref{teolevantado}. In this way we have that the lifted code $\mathfrak{L}_{\mathcal{K}_n/\cF}(C_\cL(D,G))$ is a ternary iso-dual AG-code over $\mathcal{K}$ of length $4\cdot 3^{n-m-1}$. 
Hence, we have proved the following.
	
\begin{theorem} \label{thm:ternary}
	For any $n\in \N$, the lifted code $\mathfrak{L}_{\mathcal{K}_n/\cF}(C_\cL(D,G))$, where $D$ and $G$ are as in \eqref{eq: D ter} and \eqref{eq: G ter}, is a ternary iso-dual AG-code over $\mathcal{K}_n$ of length $4\cdot 3^{n-m-1}$ and dimension $2\cdot 3^{n-m-1}$, where $m=\lceil \log_3(n) \rceil$.
\end{theorem}

\begin{remark}\label{mindT}
	The genus of $\cK_n$ for $q=2$ as well as for $q=3$ grows too quickly with respect to the length of  the lifted iso-dual codes, and then the standard estimate for the minimum distance of the lifted iso-dual codes given in  Corollary \ref{parameters} is meaningless except for very small values of $n$.  For example in the binary case, already for $n=7$ (hence $m=3$) we have that the length of the lifted iso-dual code given in Theorem \ref{thm:binary} is smaller than $2g(\cK_7)-2$. In fact, for $n=7$ the length of the lifted binary iso-dual code is $16$, and since the integers $e_1=e_7=2$, $e_2=e_6=4$, $e_3=e_5=2$ and $e_4=16$,  from \eqref{eq: gKn 2} we see that
		\[2g(\mathcal{K}_7)-2 = \tfrac{1}{2^m} \Big( 2^{n-1}(n-3)-\sum_{i=1}^{2^m-1}2^{e_i} \Big) = 
		\tfrac{1}{8} (5\cdot 2^6-32) = 28>16,\]
 as stated. A similar situation holds for the case of the lifted ternary iso-dual codes given in Theorem \ref{thm:ternary}. Therefore, since these codes are rather long, it seems an interesting problem to determine if there are alternative ways  of getting meaningful  estimates for the minimum distance of these cyclotomic iso-dual AG-codes.  
\end{remark}

 In view of Remark \ref{mindT} the authors propose the investigation of the minimum distance of the binary and ternary cyclotomic iso-dual codes.

\section{Declarations}

All authors certify that they have no affiliations with or involvement in any organization or entity with any financial interest or non-financial interest in the subject matter or materials discussed in this manuscript.

%STYLE:
%\bibliographystyle{alpha}
%\bibliographystyle{amsalpha}
\bibliographystyle{abbrv}
%\bibliographystyle{acm}
%\bibliographystyle{apalike}
%\bibliographystyle{ieeetr}
%\bibliographystyle{plain}
%\bibliographystyle{siam}
%\bibliographystyle{unsrt}

%\nocite{*}
\bibliography{Isodual} 

\end{document}